\documentclass[journal]{IEEEtran}
\usepackage{cite}
\usepackage{amsmath,amssymb,amsfonts, amsthm}
\usepackage{graphicx}
\usepackage[compatibility=false]{caption}
\usepackage{floatrow}
\usepackage{subcaption}
\usepackage{booktabs}
\usepackage{siunitx}
\usepackage{tabularx}
\usepackage{algorithm}
\usepackage{tikz}
\usetikzlibrary{shapes}
\usepackage{algpseudocode}
\usetikzlibrary{positioning}
\newcolumntype{Y}{>{\raggedright\arraybackslash}X}
\usepackage{textcomp}
\usepackage{graphicx}
\usepackage{multirow}
\usepackage{etoolbox}
\usepackage[colorlinks=true,linkcolor=blue,citecolor=blue, urlcolor=blue]{hyperref}
\usepackage{cleveref}
\makeatletter
\patchcmd{\@makecaption}
  {\scshape}
  {}
  {}
  {}
\makeatother

\newtheorem{theorem}{Theorem}
\newtheorem{lemma}{Lemma}
\newtheorem{definition}{Definition}
\newtheorem{remark}{Remark}
\newtheorem{assumption}{Assumption}
\newtheorem{property}{Property}
\newtheorem{coro}{Corollary}
\usepackage{cuted}
\usepackage[english]{babel}
\usepackage{xcolor}
    \makeatletter
 \let\old@ps@headings\ps@headings
 \let\old@ps@IEEEtitlepagestyle\ps@IEEEtitlepagestyle
 \def\confheader#1{
 \def\ps@headings{
 \old@ps@headings
 \def\@oddhead{\strut\hfill#1\hfill\strut}
 \def\@evenhead{\strut\hfill#1\hfill\strut}
 }
 \def\ps@IEEEtitlepagestyle{
 \old@ps@IEEEtitlepagestyle
 \def\@oddhead{\strut\hfill#1\hfill\strut}
 \def\@evenhead{\strut\hfill#1\hfill\strut}
 }
 \ps@headings
 }
 \makeatother
 \usepackage[pscoord]{eso-pic}

  \makeatletter
\newcommand{\linebreakand}{%
  \end{@IEEEauthorhalign}
  \hfill\mbox{}\par
  \mbox{}\hfill\begin{@IEEEauthorhalign}
}
\makeatother  
\begin{document}
\title{Adaptive Tracking Control of Euler-Lagrange Systems with Time-Varying State and Input Constraints
}
\author{Poulomee~Ghosh and Shubhendu~Bhasin
\thanks{Poulomee Ghosh and Shubhendu Bhasin are with Department of Electrical Engineering, Indian Institute of Technology Delhi, New Delhi, India. 
       {\tt\small (Email: Poulomee.Ghosh@ee.iitd.ac.in, sbhasin@ee.iitd.ac.in)}}}
\maketitle
\begin{abstract}
This paper presents an adaptive control framework for Euler–Lagrange (E–L) systems that enforces user-defined time-varying state and input constraints in the presence of parametric uncertainties and bounded disturbances. The proposed design integrates a time-varying barrier Lyapunov Function (TVBLF) with a saturated control law to guarantee constraint satisfaction without resorting to real-time optimization. A key contribution is the development of an offline, verifiable feasibility condition that certifies the existence of a feasible control policy for any prescribed pair of time-varying state and input envelopes. Additionally, we prove boundedness of all closed-loop signals. Real-time experiments conducted on a 2-DoF helicopter model validate the efficacy and practical viability of the proposed method.
\end{abstract}
\section{Introduction}
\subsection{Motivation}
In many modern safety-critical applications that can be modeled by Euler–Lagrange (E–L) dynamics, ensuring safety is not merely a design preference but a fundamental operational requirement. Small violations in safe operating conditions, whether in robotic manipulators, aerial vehicles, or surgical robots, can lead to mechanical damage, performance degradation, or even compromise human safety. Therefore, the system must remain within its admissible safe region at all times. These safety requirements can often be translated into constraints on plant states and control input, and enforcing these limits throughout the system evolution is crucial to prevent undesirable behavior. 
In a helicopter system, for example, violations of the safe operating region, whether due to higher angular positions or velocities, can result in the rotor assembly striking the frame, triggering unmodeled aerodynamic effects, or inducing unstable oscillations. Moreover, the actuators can only supply bounded voltage and torque, and exceeding these limits may result in actuator saturation and performance degradation. Therefore, constraint satisfaction is critical for safe and reliable operation, and the task becomes significantly challenging in the presence of parametric uncertainties and disturbances.
\subsection{State of the Art: Constrained Control with Static Bounds}
In existing literature, several methods have been proposed to incorporate constraints into nonlinear systems. Optimization-based techniques, such as model predictive Control (MPC) \cite{camacho2007constrained, zheng1995stability, mpc11}, control barrier function (CBF)-based approaches \cite{zcbf1, breeden2023robust, icra_cbf}, impose constraints on plant states and control input but typically demand substantial computational resources and accurate model knowledge, which are often not suitable for uncertain systems with limited computational power. Barrier Lyapunov Functions (BLFs) \cite{BLF, BLF2, salehi2020safe} provide a computationally efficient alternative by embedding state constraints directly into the stability analysis. While BLFs effectively handle static state constraints in presence of parametric uncertainties, the high control effort demanded near constraint boundaries can drive the actuator into saturation.
In practical systems, actuator saturation is often the dominant bottleneck that leads to unsafe behavior or even loss of stability \cite{li2018stability}.\\
Classical Adaptive controllers ensure asymptotic convergence of the tracking error in presence parametric uncertainties in E–L systems \cite{arc1, roy2017adaptive}, but they inherently assume unconstrained dynamics and unlimited control authority. Consequently, they can generate input that saturates the actuators or drive the system outside the safe region. In our earlier work \cite{ghoshtac, myel1}, we integrated BLF with a saturated controller in an adaptive framework to impose constraints on both state and input under parametric uncertainties and disturbances. However, that framework relied on static constraint bounds, which may become conservative when the admissible safe region evolves during operation.

\subsection{What if the admissible safe regions are time-varying?}
In real systems, the allowable state and input envelopes often vary with time due to task-dependent safety requirements, dynamic workspace conditions, evolving load dynamics, or hardware limitations that degrade over time.  During the transient phase, tracking errors and control effort are typically large, whereas both reduce significantly as the system approaches steady state. Imposing static constraints under such varying conditions either leads to unnecessary conservatism or fails to reflect the system's true safety margins. In the context of the helicopter example discussed earlier, the safe pitch and yaw limits may change over time depending on the task; precise maneuvers may require tighter bounds, while more relaxed limits can be imposed during coarse movements. Further, the available actuator torque may reduce during aggressive maneuvers due to thermal and power limitations. Predefined time-varying constraints, therefore, represent more realistic safety requirements.
\subsection{State of the Art: Constrained Control with Dynamic Bounds}
 A few recent works address time-varying state constraints. prescribed performance control (PPC) \cite{ppc1, ppc2} and funnel control (FC)-based methods \cite{fc1, fc2} explicitly constrain the output tracking error within time-varying performance bounds; however, these methods typically do not consider actuator limits. Specifically, PPC adjusts the reference trajectory online whenever input saturation occurs, which fundamentally changes the desired tracking objective \cite{ppc3}. FC approaches address this issue by relaxing the performance funnel \cite{funnelsaturation1, fc3}; however, this may compromise the user-defined performance guarantees. Time-varying BLF (TVBLF) \cite{tvblf1} provides an effective approach for shaping the error trajectories within prescribed time-varying envelopes, ensuring that the plant states evolve within the safe set as long as the desired reference remains admissible.
\subsection{Contribution}
In this work, we propose an adaptive control framework for E–L systems that enforces both time-varying state and input constraints while ensuring tracking in the presence of parametric uncertainties and bounded disturbances. The primary contributions are as follows.
\begin{enumerate}
    \item 
    We design an optimization-free adaptive controller for uncertain E-L systems with time-varying state and input constraints. 
    
    \item We derive a sufficient, offline-verifiable feasibility condition that certifies the existence of a control policy for any user-defined time-varying state and input constraints.
    
    \item We validate the proposed approach through real-time experiments on a 2-DoF helicopter platform, demonstrating that the pitch and yaw positions and velocities are constrained within predefined time-varying bounds, while the required control torques remain within time-varying input limits.
\end{enumerate}

\subsection{Notation} In this paper, $\mathbb{R}$ represents the set of real numbers, $\mathbb{R}^{p \times q}$ denotes the set of $p\times q$ real matrices, and $I_{p}$ refers to the identity matrix in $\mathbb{R}^{p \times p}$. $\log(\cdot)$ denotes the natural logarithm of $(\cdot)$, and $\|.\|$ represents the Euclidean vector norm and its induced matrix norm. $\zeta^{(i)}(t)$ signifies the $i^{th}$ derivative of $\zeta$ with respect to time. The trace of $A\in \mathbb{R}^{n \times n}$ is represented by, $\text{tr}(A)$ and the eigenvalues of A with minimum and maximum real parts are denoted by $\lambda_{min}\{A\}$ and $\lambda_{max}\{A\}$, respectively.

\section{Problem Formulation}
\label{sec:LFSR}
Consider the dynamics of an uncertain Euler–Lagrange (E–L) system given by
\begin{align}
    M(q)\ddot{q}+V_m(q,\dot{q})\dot{q}+G_r(q)+F_d(\dot{q})=\tau
    \label{plant11}
\end{align}
where $q(t), \dot{q}(t), \ddot{q}(t) \in \mathbb{R}^n$ denote the generalized position, velocity, and acceleration vectors, respectively. The matrix $M(q) \in \mathbb{R}^{n \times n}$ represents the generalized inertia, while $V_m(q,\dot{q}) \in \mathbb{R}^{n \times n}$ corresponds to the uncertain centripetal–Coriolis term. The vectors $G_r(q) \in \mathbb{R}^n$ and $F_d(\dot{q}) \in \mathbb{R}^n$ denote the uncertain generalized gravitational and frictional forces. The control input is represented by $\tau(t) \in \mathbb{R}^n$.\\
The control objective is to design a feasible control policy $\tau(t)$ such that the plant states $q(t)$, $\dot{q}(t)$ track the user-defined reference trajectories $q_d(t)$, $\dot{q}_d(t)$, respectively, while simultaneously satisfying the following constraints on the plant states and control input.

\vspace{0.3cm}
\noindent \textbf{State constraints:} The generalized position and velocity must remain within the pre-specified time-varying envelopes 
\begin{align}
    q(t)\in \Omega_q(t) \triangleq \big\{ q \in \mathbb{R}^n \mid \|q\| < \phi_q(t)\big\} \qquad \forall t \ge 0\label{state_con_el_1}\\
    \dot{q}(t)\in\Omega_{\dot{q}}(t) \triangleq \big\{ \dot{q} \in \mathbb{R}^n\mid \|\dot{q}\| < \phi_{\dot{q}}(t) \big\} \qquad \forall t \ge 0
    \label{state_con_el_2}
\end{align}
where, $\phi_{q},\;\phi_{\dot{q}}\colon [0,\infty)\to (0, \infty)$ are continuously differentiable, user-defined positive functions.

\vspace{0.3cm}
\noindent \textbf{Input constraint:} The control input must satisfy
\begin{align}
\tau(t)\in\Omega_{\tau}(t) \triangleq \big\{ \tau \in \mathbb{R}^n \mid\|\tau\| \leq \phi_{\tau}(t) \big\} \qquad \forall t \ge 0
\label{input_con_el}
\end{align}
where, $\phi_{\tau} :[0,\infty)\rightarrow (0, \infty) $ is a continuously differentiable user-defined positive function.\\
For subsequent controller design and stability analysis, the following standard properties of E–L dynamics are assumed to hold \cite{Vidyasagar}:
\begin{property}
\label{prop_sym_pd}
The inertia matrix $M(q)$ is symmetric and positive-definite, satisfying
\begin{align}
    k_{m_{1}} \|\mu\|^2 \leq \mu^T M(q) \mu \leq k_{m_2} \|\mu\|^2
\end{align}
for all $\mu \in \mathbb{R}^n$, where $k_{m_{1}}$ and $k_{m_{2}}$ are known positive constants.
\end{property}
\begin{property}
\label{prop_bound}
    The centripetal-Coriolis matrix $V_m(q,\dot{q})$, the gravity vector $G_r(q)$ and the friction vector $F_d(\dot{q})$ satisfy 
    \begin{align}
        & \|V_m(q,\dot{q})\|\leq k_v\|\dot{q}\|\\
        &\|G_r(q)\|\leq k_g\\
        & \|F_d(\dot{q})\|\leq k_{f_1}+k_{f_2}\|\dot{q}\|
    \end{align}
    where $k_v$, $k_g$, $k_{f_{1}}$ and $k_{f_{2}}$ are known positive constants. 
\end{property}
\begin{remark}
\label{remark_known_bounds}
For E–L systems, the bounds, $k_{m_2}, k_v, k_g, k_{f_1}$, and $k_{f_2}$ depend on physical parameters such as link lengths, masses, inertial properties, etc. which are typically available from system specifications or can be identified through standard parameter estimation or experimental calibration.
\end{remark}
\begin{property}
\label{prop_skew}
The matrix $\dot{M}(q) - 2 V_m(q,\dot{q})$ is skew-symmetric, i.e.,
\begin{align}
    \mu^T \big( \dot{M}(q) - 2 V_m(q,\dot{q}) \big) \mu = 0, \quad \forall \mu \in \mathbb{R}^n.
\end{align}
\end{property}

\begin{property}
\label{prop_lip}
The system is linearly parameterizable, i.e.,
\begin{align}
    Y_1(q,\dot{q},\ddot{q}) \theta = M(q) \ddot{q} + V_m(q,\dot{q}) \dot{q} + G_r(q) + F_d(\dot{q})
\label{Property4}
\end{align}
where $Y_1: \mathbb{R}^n \times \mathbb{R}^n \times \mathbb{R}^n \rightarrow \mathbb{R}^{n \times m}$ is the known regression matrix and $\theta \in \mathbb{R}^m$ is the unknown constant parameter vector.
\end{property}
\section{Proposed Methodology}
\label{PM}
The position and velocity tracking errors are defined as 
\begin{align}
    e\triangleq q-q_d\\
    \dot{e}\triangleq \dot{q}-\dot{q}_d
\end{align}
To facilitate the constrained control design, we introduce the filtered tracking error
\begin{align}
    r \;\triangleq\; \dot e + \alpha\, e 
    \label{fte}
\end{align}
where $\alpha\in \mathbb{R}$ is a known positive constant.
\subsection{Constraint Transformation}
We provide details of how constraints on plant states are transformed into constraints on tracking errors.
\begin{assumption}
\label{ref_model_assumption}
  The reference trajectory $q_d(t)$, its first and second derivatives, $\dot{q}_d(t)$ and $\ddot{q}_d(t)$, respectively, are bounded such that
\begin{align}
    & \|q_d(t)\| \leq \phi_{q_d}(t) < \phi_q(t) && \forall t \geq 0 \label{assump1a}\\
    & \|\dot{q}_d(t)\|\leq \phi_{\dot{q}_d}(t) < \phi_{\dot{q}}(t) && \forall t \geq 0 \label{assump1b}\\
    & \|\ddot{q}_d(t)\|\leq \phi_{\ddot{q}_d}(t) && \forall t \geq 0 \label{assump1c}
    \end{align}
    where $\phi_{q_d},\phi_{\dot{q}_d}, \phi_{\ddot{q}_d}:[0,\infty)\rightarrow (0, \infty )$ are continuously differentiable known positive functions.
\end{assumption}
Provided Assumption~\ref{ref_model_assumption} holds, the state constraints can be transformed to constraints on the tracking errors
\begin{align}
    &e(t)\in\Omega_e(t)\triangleq\{e\in\mathbb{R}^n\mid\|e\|<\phi_e(t)\} \qquad \forall t\ge 0 \label{error_con}\\
    &\dot{e}(t)\in\Omega_{\dot{e}}(t)\triangleq\{\dot{e}\in\mathbb{R}^n\mid \|\dot{e}\|<\phi_{\dot{e}}(t)\} \qquad \forall t\ge 0 \label{errordot_con}
\end{align}
where, $\phi_e,\phi_{\dot e}:[0,\infty)\to(0,\infty)$ are continuously differentiable functions defined as
\begin{align}
    &\phi_e(t)\triangleq \phi_q(t)-\phi_{q_d}(t) \quad \forall t\ge 0\\
    &\phi_{\dot{e}}(t)\triangleq\phi_{\dot{q}}(t)-\phi_{\dot{q}_d}(t) \quad \forall t \ge 0
\end{align}
Satisfying \eqref{error_con} and \eqref{errordot_con} imply satisfying \eqref{state_con_el_1} and \eqref{state_con_el_2}, respectively.\\

\begin{assumption}
Initially, the trajectory tracking and filtered tracking errors satisfy the following bounds
\begin{align}
   &\|e(0)\|\leq \phi_e(0)\label{ezero}\\
   &\|{r}(0)\|<\phi_r(0)
\end{align}
\label{error_assumption}
\end{assumption}
\begin{remark}
    Assumption~\ref{error_assumption} implies that the plant states initially remain within the user-defined safe set to prevent constraint violation at the initial time, which is a standard assumption in BLF-based constrained control design \cite{BLF, BLF2, elblf1}. 
\end{remark}
\begin{lemma}
\label{constraint_conv_lemma}
    [Constraint conversion from $(e,\dot e)$ to $r$]\label{lemma_constraint_conv} Let  $\alpha$ satisfy the gain condition 
    \begin{align}
    \alpha<\frac{\phi_{\dot{e}}(t)}{\phi_{e}(t)} \qquad\forall t \ge 0
    \label{alpha}
    \end{align} 
    and suppose there exists a continuously differentiable function $\phi_r:[0,\infty)\rightarrow(0,\infty)$, which satisfies the following inequality
\begin{align}
 \phi_r(t)<\min\Big\{\dot \phi_e(t)+\alpha\phi_e(t),\ \ \phi_{\dot e}(t)-\alpha\phi_e(t)\Big\} 
    \label{phir_1}
\end{align}
Provided Assumption~\ref{error_assumption} holds and the filtered tracking error remains bounded, i.e.,
\begin{align}
    r(t)\in\Omega_r(t)\triangleq\{r\in\mathbb{R}^n \mid\|r\|<\phi_r(t)\} \qquad \forall t\ge 0
\end{align}
the trajectory tracking errors remain within the time-varying envelopes, i.e., $ e(t)\in \Omega_e(t)$, $\dot{e}(t)\in\Omega_{\dot{e}}(t)$ for all $t\ge 0$.
\end{lemma}
\begin{proof}
Let $w(t)\triangleq \|e(t)\|$. For $e(t)\neq 0$,
\begin{align}
    \dot w(t)
    = \frac{e^\top \dot e}{\|e\|}
    = \frac{e^\top (r-\alpha e)}{\|e\|}
    \le \|r\| - \alpha\,\|e\|
    = \|r\| - \alpha\,w.
    \label{eq:w_dot_ineq}
\end{align}
At points where $e(t)=0$, the upper Dini derivative satisfies $D^+ w \le \|\dot e\| \le \|r\|$; so, \eqref{eq:w_dot_ineq} still holds in the sense of an upper bound.\\
\textit{Step 1 (bound on $\|e\|$):}
Using \eqref{phir_1} and \eqref{eq:w_dot_ineq},
\begin{align}
    \dot w(t) < \phi_r(t) - \alpha\,w(t)
    \label{eq:w_comp}
\end{align}
Consider the scalar comparison system $\dot z = \phi_r - \alpha z$ with $z(0)=\phi_e(0)$. Using \eqref{phir_1},
\begin{align}
    \dot \phi_e(t) > \phi_r(t) - \alpha\,\phi_e(t)
\end{align}
hence $\phi_e(\cdot)$ is a supersolution of \eqref{eq:w_comp}. Since $w(0)\le \phi_e(0)$ by \eqref{ezero}, the comparison lemma yields $w(t)< \phi_e(t)$ $\forall t\ge 0$, which implies $\|e(t)\|< \phi_e(t)$ $\forall t\ge 0$.\\
\textit{Step 2 (bound on $\|\dot e\|$):}
Using $\dot e = r - \alpha e$, for all $t \ge 0$
\begin{align}
    \|\dot e(t)\| \le \|r(t)\| + \alpha\,\|e(t)\|
    < \phi_r(t) + \alpha\,\phi_e(t)
    < \phi_{\dot e}(t)
\end{align}
This establishes $\dot{e}(t)\in \Omega_{\dot{e}}(t)$ $\forall t \ge 0$.
\end{proof}
\subsection{Error Dynamics}
Now, using the E–L dynamics in \eqref{plant11} and Property~\ref{prop_lip}, the derivative of filtered error dynamics can be expressed as
\begin{align}
    M\dot r = Y\theta + \tau + V_m r 
    \label{eq:r_dynamics}
\end{align}
where, $Y\in\mathbb{R}^{n\times m}$ is the known regressor matrix, $\theta\in\mathbb{R}^m$ is the unknown constant parameter vector and $Y\theta$ can be expressed as
\begin{equation}
\label{eq:Ytheta}
Y\theta
= M\big(\alpha\,\dot e - \ddot q_d\big)
  + V_m\big(r - \dot q\big)
  - F_d - G_r
\end{equation}
\begin{assumption}
\label{theta_assumption}
    The norm of the unknown plant parameter vector is bounded such that $\|\theta\|<\bar{\theta}$, where $\bar{\theta}>0$ is assumed to be known.
\end{assumption}
\begin{remark}
Assumption~\ref{theta_assumption} is commonly used in projection-based adaptive control literature \cite{arc1}, where adaptive parameters are bounded within a predefined set, avoiding parameter drift in the presence of disturbances. In practical systems, unknown parameters are often naturally bounded due to physical constraints or modeling assumptions, and therefore, it is reasonable to assume the bound  \(\bar{\theta}\). Moreover, by incorporating system-specific information, \(\bar{\theta}\) can be estimated without being overly conservative. 
\end{remark}

\subsection{Saturated Control Design for Time-Varying Input Constraints}
\label{input_con_subsection}
To bound the control input magnitude within a time-varying envelope, we design an auxiliary control input $\tau_a(t)\in \mathbb{R}^n$ as
\begin{align}
&\tau_a=-Y\hat{\theta}-Kr+\frac{\dot{\phi_{r}}}{\phi_{r}}k_{m_{2}}r
\label{pc1}
\end{align}
where, $\hat{\theta}(t)\in\mathbb{R}^m$ is the estimate  of the unknown parameter vector and $K\in \mathbb{R}^{n \times n}$ is a positive parameter gain.
To enforce the input constraint, we design a saturated controller
\begin{align}
&\tau(t)= \begin{cases}
\tau_a(t) & \text{if}\:\:\: \|\tau_a(t)\|\leq \phi_{\tau}(t)\\
\dfrac{\phi_{\tau}(t)}{\|\tau_a(t)\|}\tau_a(t) & \text{if}\:\:\: \|\tau_a(t)\|>\phi_{\tau}(t)
\end{cases}
\label{pc2}
\end{align}
Using (\ref{pc1}) and (\ref{pc2}), the closed-loop dynamics of the filtered tracking error can be expressed as
\begin{align}
M\dot{r}=Y\tilde{\theta}-Kr-V_mr+\frac{\dot{\phi_{r}}}{\phi_{r}}k_{m_{2}}r+\Delta \tau
    \label{edot}
\end{align}
where $\tilde{\theta}(t)\triangleq\theta(t)-\hat{\theta}(t)\in \mathbb{R}^{m}$ is the parameter estimation error and $\Delta \tau(t)\in\mathbb{R}^{m}$ is defined as the difference between the actual control input $\tau(t)$ and auxiliary control input $\tau_a(t)$, i.e., 
\begin{align}
    \Delta \tau(t) \triangleq \tau(t)-\tau_a(t)
\end{align}
where,
\begin{align}
\Delta\tau(t) 
= \begin{cases}
0, & \|\tau_a(t)\|\le \phi_\tau(t),\\[1mm]
\Big(1-\dfrac{\phi_\tau(t)}{\|\tau_a(t)\|}\Big)\tau_a(t), & \|\tau_a(t)\|> \phi_\tau(t),
\end{cases}
\label{delta_tau}
\end{align}
The term `$\Delta\tau(t)$' captures the effect of input saturation; it vanishes when $\|\tau_a(t)\|\le\phi_\tau(t)$
 and otherwise acts as a bounded perturbation, which can further be written as 
\begin{align}
    \|\Delta \tau\|\leq\|\tau_a\|-\phi_{\tau}
    \label{Delta_tau_1}
\end{align}
where,
\begin{align}
    \|\tau_a\|&\leq \|Y \hat{\theta}\|+\|Kr\|+\|\dfrac{\dot{\phi}_r} 
    {\phi_r}k_{m_{2}}r\|
    \label{taua_bound}
    \end{align}
    Employing Property~\ref{prop_bound}, \eqref{taua_bound} can be expanded as
    \begin{align}
    \|\tau_a\|&\leq \bar{\theta}(k_{m_{2}}(\alpha \|\dot{e}\|+\|\ddot{q}_d\|)+k_v\|\dot{q}\|(\|r\|+\|\dot{q}\|)+k_g \nonumber \\ &+k_{f_{1}}+k_{f_{2}}\|\dot{q}\|)+\lambda_{max}\{K\}\|r\| + \dfrac{|\dot{\phi}_r|}{\phi_r}k_{m_{2}}\|r\|
    \label{vbound}
\end{align}
Now, solving the differential equation (\ref{fte}) and employing Assumption \ref{error_assumption}, we obtain  $\|e(t)\|\leq \|e(0)\|+\frac{\|r(t)\|}{\alpha}$ and $\|\dot{e}(t)\|\leq \alpha\|e(0)\|+2\|r(t)\|$ $\forall t\geq 0$, which in turn implies $\|\dot{q}(t)\|\leq \alpha \|e(0)\|+2\|r(t)\|+\phi_{\dot{q}}(t)$ $\forall t \geq 0$. 
substituting the bounds in (\ref{vbound}),
\begin{align}
    \|\tau_a\| 
     &\leq \Psi_1 \|r\|^2 + (\Psi_2 +\frac{|\dot{\phi}_r|}{\phi_r} k_{m_2})\|r\| + \Psi_3
     \label{tau_bound}
 \end{align}
where $\Psi_1 = 6\bar{\theta}k_v$ is a known positive constant and  $\Psi_2(t)= 
\bar{\theta}(2\alpha k_{m_2} +5k_v\big(\alpha\phi_e(0)+\phi_{\dot q_d}(t)\big) +2k_{f_2})+\lambda_{\max}\{K\}$, $\Psi_3(t)=\bar{\theta}[k_{m_2}\big(\alpha^{2}\phi_e(0)+\phi_{\ddot{q}_d}(t)\big)
+k_v\big(\alpha\phi_e(0)+\phi_{\dot q_d}(t)\big)^{2}
+k_g + k_{f_1}
+k_{f_2}\big(\alpha\phi_e(0)+\phi_{\dot q_d}(t)\big)]$ are known positive functions. Substituting the bounds in (\ref{Delta_tau_1}),
\begin{align}
    \|\Delta \tau\|\leq \Psi_1\|r\|^2+(\Psi_2+\frac{|\dot{\phi}_r|}{\phi_r} k_{m_2})\|r\|+\Psi_3-\phi_{\tau}
    \label{deltatau_bound}
\end{align}

\subsection{TVBLF-based Adaptive Law Design for Time-Varying State Constraints}\label{State_constraint_TVBLFEL}
To bound plant states within an pre-defined time-varying envelope, we employ a TVBLF \cite{tee2011control}, defined as follows.
\begin{figure}
    \centering
\includegraphics[width=0.8\linewidth]{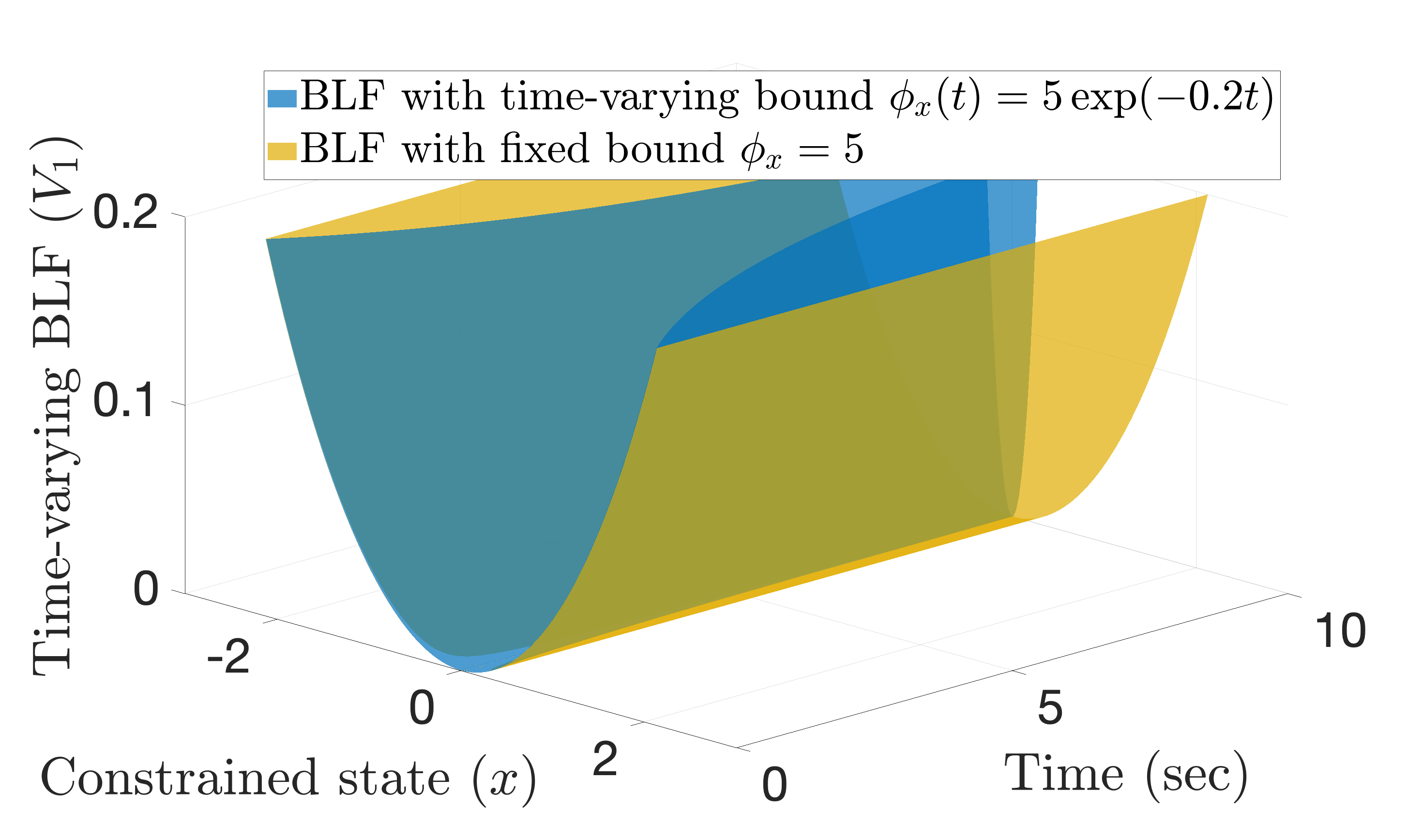}
    \caption{BLF \eqref{tvblf_def} with fixed and time-varying bounds.}
    \label{tvblf_figure}
\end{figure}
\begin{definition}[TVBLF]
Let \(h:\mathbb{R}^n\times\mathbb{R}_{\ge 0}\to\mathbb{R}\) be a continuously differentiable function, and define the open time-varying region
\begin{align}
\Omega_x(t)=\{x\in\mathbb{R}^n \mid h(x,t)>0\}
\end{align}
which contains the origin for all \(t\ge 0\). A TVBLF is a scalar function \(V(x,t)\), defined with respect to the system \(\dot{x}=f(x,t)\) on \(\Omega_x(t)\), that is continuous and positive definite in \(x\) for all \(t\ge 0\), has continuous first-order partial derivatives with respect to \(x\) and \(t\) at every point of \(\Omega_x(t)\), possesses the barrier property that \(V(x,t)\to\infty\) as \(x\) approaches the boundary of \(\Omega_x(t)\) (equivalently, as \(h(x,t)\to 0^+\)), and satisfies \(V(x(t),t)\le b\) for all \(t\ge 0\) along the solution of \(\dot{x}=f(x,t)\), for \(x(0)\in\Omega_x(0)\) and some positive constant \(b\).
\end{definition}
A commonly used TVBLF is the logarithmic barrier function of the form
\begin{equation}
V_1= \log\frac{\phi_x(t)^2}{\phi_x(t)^2-x^{\top}x}
\label{tvblf_def}
\end{equation}
defined on \( \Omega_x(t) =\{x\in\mathbb{R}^n\mid\|x\|<\phi_x(t)\}\), where  $\phi_{x}:[0,\infty)\to(0,\infty)$ is continuously differentiable positive function. As the state \(\|x\|\) approaches the boundary $\phi_x(t)$, the TVBLF diverges, i.e., $V_1\rightarrow \infty$. Fig.~\ref{tvblf_figure} shows the evolution of the safe sets in case of BLF \eqref{tvblf_def} with fixed and time-varying bounds.\\
Now, to constrain the filter tracking error $r(t)$, which in turn ensures tracking error and state constraint satisfaction, we consider the following TVBLF $V_r$ defined on the set $\Omega_r^{'}(t):\{r\in\mathbb{R}^n\mid\:k_{m_{2}}\|r\|^2<\phi_r^{'2}(t)\}$
\begin{align}
    &V_r\triangleq\frac{1}{2}\log \frac{\phi_r^{'2}}{\phi_{r}^{'2}-r^{\top}Mr}
\end{align}
where $\phi^{'}_r=\phi_r\sqrt{k_{m_{2}}}$. Since $r^{\top}Mr\leq k_{m_{2}}\|r\|^2$ (Property~\ref{prop_sym_pd}), $k_{m_{2}}\|r\|^2<\phi_r^{'2}(t)\implies r^{\top}Mr<\phi_r^{'2}(t)$ $\forall t \ge 0$. Further, as $r^{\top}Mr\rightarrow \phi_r^{'2}(t)$, i.e.,
when the constrained state $r(t)$ approaches the boundary of the safe set, the BLF $V_r\rightarrow \infty$.
For the proposed controller (\ref{pc1}) and (\ref{pc2}), the following adaptive update laws are defined, which will be employed in the subsequent Lyapunov analysis.
\begin{align}
    &\dot{\hat{\theta}}=\text{proj}_{\Omega_{\theta}}\bigg(\frac{\Gamma Y^Tr}{\phi_{r}^{'2}-k_{m_{2}}\|r\|^2}\bigg)  
    \label{updatelaw}
\end{align}
where $\Gamma\in\mathbb{R}^{m\times m}$ is a user-defined positive-definite matrix. The projection operator \cite{lavretsky2012robust}, denoted as \(\text{proj}(\cdot)\), ensures that parameters remain bounded within a convex and compact region in the parameter space defined by \(\Omega_{\theta}=\{\hat{\theta}\in\mathbb{R}^m|\|\hat{\theta}\|^2 \leq \bar{\theta}^2\}\). 
\begin{theorem}
\label{maintheorem_tvblf_el}
For the E-L system (\ref{plant11}), provided Assumptions~\ref{ref_model_assumption}-\ref{theta_assumption} and the gain condition \eqref{alpha} hold, the controller (\ref{pc1}), (\ref{pc2}) and the adaptive update law (\ref{updatelaw}) ensure closed-loop stability of the tracking error dynamics, time-varying state and input constraint satisfaction, i.e., $q(t),\dot{q}(t)\in \Omega_q(t)$, $\tau(t)\in \Omega_{\tau}(t)$, $\forall t \ge 0$ and boundedness of all closed-loop signals,  subject to the following feasibility condition (C1).\\
\textbf{C1:} The bound on the control input satisfies the following inequality,
\begin{align}
   \phi_{\tau}(t)>(\Psi_1\phi_r(t)+\Psi_2^{'}(t)+|\dot{\phi}_r(t)|k_{m_{2}})\phi_r(t)+\Psi_3(t)
    \label{f1}
\end{align}
where $\Psi_2(t)=\Psi_2(t)-\lambda_{min}\{K\}$ and $\Psi_1,\Psi_2(t),\Psi_3(t)$ are defined in Section~\ref{input_con_subsection}. 
\end{theorem}

\begin{proof}
Consider the candidate Lyapunov function $V:\Omega_r^{'}\times \mathbb{R}^{N}\rightarrow \mathbb{R}_{+}$ as,
\begin{align}
    V&=
    \frac{1}{2}
    \bigg[\log\frac{\phi_r^{'2}}{\phi_{r}^{'2}-r^{\top}Mr}+\tilde{\theta}^T\Gamma^{-1}\tilde{\theta}\bigg]
    \label{lyap}
\end{align}
Taking the time derivative of $V$ along the system trajectory 
\begin{align}
    \dot{V}=&\frac{1}{2(\phi_{r}^{'2}-r^{\top}Mr)}\bigg[2r^{\top}(Y\tilde{\theta}-Kr-V_mr+\frac{\dot{\phi_{r}}}{\phi_{r}}k_{m_{2}}r\nonumber\\
    &+\Delta\tau)+r^{\top}\dot{M}r\bigg]-\frac{\dot{\phi}_r r^{\top}Mr}{\phi_r(\phi_r^{'2}-r^{\top}Mr)}-\tilde{\theta}^{\top}\Gamma^{-1}\dot{\hat{\theta}}
    \label{vdot111_el}
    \end{align}
    Employing Property 2 and the adaptive update laws in (\ref{updatelaw}),
\begin{align}
    \dot{V}\leq&-\frac{r^TKr}{\phi_{r}^{'2}-k_{m_{2}}\|r\|^2}+\frac{r^T\Delta \tau}{\phi_{r}^{'2}-k_{m_{2}}\|r\|^2}
    \label{vdot2_EL}
\end{align}
 Substituting the bound of $\|\Delta \tau(t)\| $ from \eqref{deltatau_bound} in (\ref{vdot2_EL}),
\begin{align}
    \dot{V}\leq &-\frac{\lambda_{min}\{K\}\|r\|^2}{\phi_{r}^{'2}-k_{m_{2}}\|r\|^2}
    \nonumber\\
    & +\frac{ (\Psi_1\|r\|^2+(\Psi_2+\frac{|\dot{\phi}_r|}{\phi_r})\|r\|+\Psi_3-\phi_{\tau})\|r\|}{\phi_{r}^{'2}-k_{m_{2}}\|r\|^2} \\
    \leq &-\frac{\lambda_{min}\{K\}\|r\|}{\phi_{r}^{'2}-k_{m_{2}}\|r\|^2}\bigg [\frac{ \phi_{\tau}-\Psi_3}{\lambda_{min}\{K\}}\nonumber\\
     &-\bigg(\frac{\Psi_1\|r\|+\Psi_2+\frac{|\dot{\phi}_r|}{\phi_r}k_{m_{2}}}{\lambda_{min}\{K\}}-1\bigg)\|r\|\bigg]
    \label{vdot112}
\end{align}
Now, invoking Assumption~\ref{error_assumption}, i.e., \(\|r(0)\|<\phi_r(0)\)  which implies
\(k_{m_{2}}\|r(0)\|^2<\phi_r^{'2}(0)\), to achieve a stability result from (\ref{vdot112})  the following  condition must hold.
\begin{align}
    \phi_{\tau}>(\Psi_1\phi_r+\Psi_2^{'}+|\dot{\phi}_r|k_{m_{2}})\phi_r+\Psi_3
    \label{c11}
\end{align}
which is the feasibility condition C1, where $\Phi_2^{'}(t)=\Psi_2(t)-\lambda_{min}\{K\}$ is a known positive function. Provided C1 is satisfied, since \(k_{m_{2}}\|r(0)\|^2<\phi_r{'^2}\), employing Lemma 1, at $t=0$, (\ref{vdot112}) can be written as,
\begin{align}
   \dot{V}(0)\leq0
   \label{uub1}
\end{align}
Now, define $t^\star\triangleq\inf\{t>0:\; k_{m_{2}}\|r(t)\|^2=\phi_r^{'2}(t)\}$. Provided Assumption~\ref{error_assumption} and C1 hold
$r(t)\in\Omega_r$ for all $t\in[0,t^\star)$ and $\dot{V}(t)\le0$ for all \(t\in[0,t^\star)\).
Hence \(V(t)\le V(0)<\infty\) on \([0,t^\star)\). Consequently, if $t^\star<\infty$, then as $t\rightarrow t^\star$ the BLF term $V_r(t)\rightarrow\infty$, and therefore $V(t)\to\infty$,
contradicting the bound \(V(t)\le V(0)<\infty\) on \([0,t^\star)\).
Therefore, \(t^\star=\infty\) it can be inferred that
\begin{align}
k_{m_{2}}\|r(t)\|^2<\phi_r^{'2}(t)\implies \|r(t)\|< \phi_r (t)\hspace{10pt} \forall t\geq0
  \label{lambda1}
\end{align}
Invoking Lemma~\ref{constraint_conv_lemma}, it can be proved that 
\begin{align}
    &\|e(t)\|< \phi_e(t)\hspace{20pt} \forall t\geq0 \label{ebound_lyap_el}\\
    &\|\dot{e}(t)\|<\phi_{\dot{e}}(t)\hspace{20pt} \forall t\geq0 \label{edot_bound_lyap_el}
\end{align} 
i.e., the trajectory tracking error and its derivative will be constrained within the user-defined safe sets.\\
Furthermore, since the desired trajectory and its derivative are bounded, i.e., $\|q_d(t)\|\leq \phi_{q_{d}}(t)$, $\|\dot{q}_d(t)\|< \phi_{\dot{q}_d}(t)$, it follows directly from (\ref{ebound_lyap_el}) and \eqref{edot_bound_lyap_el} that the proposed controller ensures the plant states remain within the pre-defined time-varying envelopes
\begin{subequations}
\begin{align}
    &\|q(t)\|<\phi_e(t)+\phi_{q_{d}}(t)= \phi_q(t)  && \forall t \geq 0\\
     &\|\dot{q}(t)\|<\phi_{\dot{e}}(t)+\phi_{\dot{q}_d}(t)= \phi_{\dot{q}}(t)&& \forall t \geq 0
\end{align}
\end{subequations}
Since the closed-loop trajectory tracking error and the parameter estimation errors are bounded and $\theta$ is constant, it follows that the estimated parameters are also bounded, i.e., $\hat{\theta}(t)\in \mathcal{L}_{\infty}$. Consequently, the plant states $q(t), \dot{q}(t)$ and control input $\tau(t)$ remain bounded at all times. Thus, the proposed controller guarantees that all the closed-loop signals are bounded.
\end{proof}
\begin{remark}
The feasibility condition (C1) captures the interrelationships and associated trade-offs between the prescribed time-varying error constraints, the reference trajectory, and the admissible input constraints. Particularly, C1 ensures that the control authority available through the input constraint \( \phi_\tau(t) \) should be sufficient to handle the prescribed dynamic bounds on the plant states (consequently, on the errors) imposed by \( \phi_r(t) \), despite the presence of system nonlinearities, parametric uncertainties, and external disturbances. The sufficient feasibility condition (C1) also shows that overly restrictive state constraints or excessively fast convergence requirements may render the control problem infeasible under a given input bound. Conversely, for a fixed input constraint envelope, C1 provides a guideline for selecting admissible performance functions and subsequently, time-varying constraints.
\end{remark}
Algorithm~\ref{algo1} provides a guideline to check the feasibility of any given time-varying constraint set $\{\phi_e(t),\phi_{\dot{e}}(t),\phi_{\tau}(t)\}$, where $\epsilon_1,\epsilon_2\in\mathbb{R}_{> 0}$ are user-defined constants.
\begin{algorithm}[h!]
\caption{Feasibility Check}
\textbf{Input:} (i) Time varying state and input constraints: $\phi_e(t)$, $\dot{\phi}_e(t)$, $\phi_{\dot{e}}(t)$, $\phi_{\tau}(t)$; (ii) Known parameters and functions: $\Psi_1$, $\Psi_2^{'}(t)$, $\Psi_3(t)$, $k_{m_{2}}$; (iii) Design parameters: $K$, $\epsilon_1$, $\epsilon_2$.\\
\textbf{Output:} 
Feasibility or infeasibility of constraint set $\{\phi_{e},\phi_{\dot{e}}, \phi_{\tau}\}$.  

\begin{algorithmic}[1]
\State Compute \( \alpha_{\max}(t)=\frac{\phi_{\dot e}(t)}{\phi_e(t)} \)
\State Compute \( \underline{\alpha}_{\max}=\inf_{t\ge 0}\alpha_{\max}(t) \)  
\State Select \( \alpha \leftarrow \underline{\alpha}_{\max}-\epsilon_1 \)  

\State Compute 
\[\phi_{r_{max}}(t)=\min\Big\{\dot \phi_e(t)+\alpha\phi_e(t),\ \ \phi_{\dot e}(t)-\alpha\phi_e(t)\Big\}\]
\State Select $\phi_r(t)\leftarrow\phi_{r_{max}}(t)-\epsilon_2$
    \If{$\phi_{\tau}<(\Psi_1\phi_r+\Psi_2^{'}+|\dot{\phi}_r|k_{m_{2}})\phi_r+\Psi_3$} \hspace{2pt} 
    \State \Return{The constraint set  is infeasible (C1 violated)}
    \Else $\;${The constraint set is feasible}
    \EndIf
\end{algorithmic}
\label{algo1}
\end{algorithm}
\begin{remark}
    The feasibility condition C1 is a sufficient offline certificate derived under worst-case assumptions, including the maximum possible saturation error. Consequently, when the actual saturation effects are less severe than their worst-case bounds, C1 may yield a more conservative error envelope than what is strictly necessary. This conservatism trades off offline verifiability, which guarantees forward invariance of the prescribed constraint sets while preserving the original reference trajectory and the performance envelop. 
\end{remark}
\begin{coro}
\label{coro:corollary_vanishing_saturation_error_tvblf_el}
(\textit{Convergence under vanishing saturation error}) Suppose the saturation error vanishes, i.e., there exists $t_1\in[0,\infty)$ such that $\Delta \tau(t)=0$ $\forall t\ge t_1$. Then, asymptotic tracking is recovered, $e(t)\to 0$ and $\dot{e}(t)\to 0$ as $t\to\infty$.
\end{coro}
\begin{proof}
For $t\ge t_1$, with $\Delta \tau(t)=0$, \eqref{vdot2_EL} reduces to
\begin{align}
\dot V \le -\frac{\lambda_{\min}\{K\}\|r\|^2}{\phi_r^{'2}-k_{m_{2}}\|r\|^2}\le0
\label{vdot_coro_tvblfel}
\end{align}
From \eqref{vdot_coro_tvblfel}, we can show that $r(t)\in\mathcal{L}_2$ and from \eqref{edot} we obtain $\dot{r}(t)\in \mathcal{L}_{\infty}$, which proves $r(t)$ is uniformly continuous on $[t_1,\infty)$. Consequently,  Barbalat’s lemma \cite{slotine} yields  $\lim_{t\to\infty} r(t)=0$ and from \eqref{fte} it can be obtained that $e(t)\to 0$ and $\dot{e}(t)\to 0$ as $t\to 0$, which imply $q(t)\to q_d(t)$ and $\dot{q}(t)\to \dot{q}_d(t)$ as $t\to \infty$, respectively.
\end{proof}
\begin{coro}
(\textit{Plant under disturbance}) Consider the uncertain E-L plant with an unknown bounded disturbance.
\begin{align}
     M(q)\ddot{q}+V_m(q,\dot{q})\dot{q}+G_r(q)+F_d(\dot{q})=\tau+d
     \label{plant_disturbance_el}
\end{align}
where the disturbance $d(t)\in\mathbb{R}^n$ is bounded such that $\sup_{t\ge0}\|d(t)\|<\bar{d}$ and $\bar{d}>0$ is a known constant. Under the projection-based adaptive law \eqref{updatelaw},
the time-varying state and input constraints \eqref{state_con_el_1}--\eqref{input_con_el} are satisfied,
and all closed-loop signals remain  bounded for all $t\ge0$ provided the following
modified feasibility condition holds for all $t\ge0$:
\begin{equation}
\label{C1d_el}
\phi_{\tau}>(\Psi_1\phi_r+\Psi_2+|\dot{\phi}_r|k_{m_{2}}-\lambda_{min}\{K\})\phi_r+\Psi_3^{'}
\end{equation}
where, $\Psi_3^{'}=\Psi_3+\bar{d}$ is a known positive constant.
\end{coro}
\begin{proof}
   In presence of disturbance, \eqref{vdot2_EL} can be rewritten as
    \begin{align}
        \dot{V}\leq&-\frac{r^TKr}{\phi_{r}^{'2}-k_{m_{2}}\|r\|^2}+\frac{r^T\Delta \tau}{\phi_{r}^{'2}-k_{m_{2}}\|r\|^2}+\frac{r^Td}{\phi_{r}^{'2}-k_{m_{2}}\|r\|^2}
    \label{vdot2_EL_dis}
    \end{align}
    Using 
$r^T d \le \bar d\,\|r\|$ and the bound of $\|\Delta \tau(t)\|$ in \eqref{deltatau_bound}
yields the inequality
    \begin{align}
    \dot{V}\leq &-\frac{\lambda_{min}\{K\}\|r\|}{\phi_{r}^{'2}-k_{m_{2}}\|r\|^2}\bigg [\frac{ \phi_{\tau}-\Psi_3^{'}}{\lambda_{min}\{K\}}\nonumber\\
     &-\bigg(\frac{\Psi_1\|r\|+\Psi_2+\frac{|\dot{\phi}_r|}{\phi_r}}{\lambda_{min}\{K\}}-1\bigg)\|r\|\bigg]
    \label{vdot112_el_dis}
    \end{align}
where, $\Psi_3^{'}=\Psi_3+\bar{d}$ and the modified feasibility condition becomes,
\begin{align}
    \phi_{\tau}>(\Psi_1\phi_r+\Psi_2+|\dot{\phi}_r|k_{m_{2}}-\lambda_{min}\{K\})\phi_r+\Psi_3^{'}
    \label{c11_dis_el}
\end{align}
If \eqref{c11_dis_el} is satisfied, \eqref{uub1} holds, and it can be followed from the proof of Theorem~\ref{maintheorem_tvblf_el} that the time-varying state and input constraints \eqref{state_con_el_1}--\eqref{input_con_el}  are satisfied and all the closed-loop signals remain bounded for all time.
\end{proof}
\begin{figure*}[t]
    \centering
    \begin{subfigure}[t]{0.48\textwidth}
    \centering
\includegraphics[width=\linewidth]{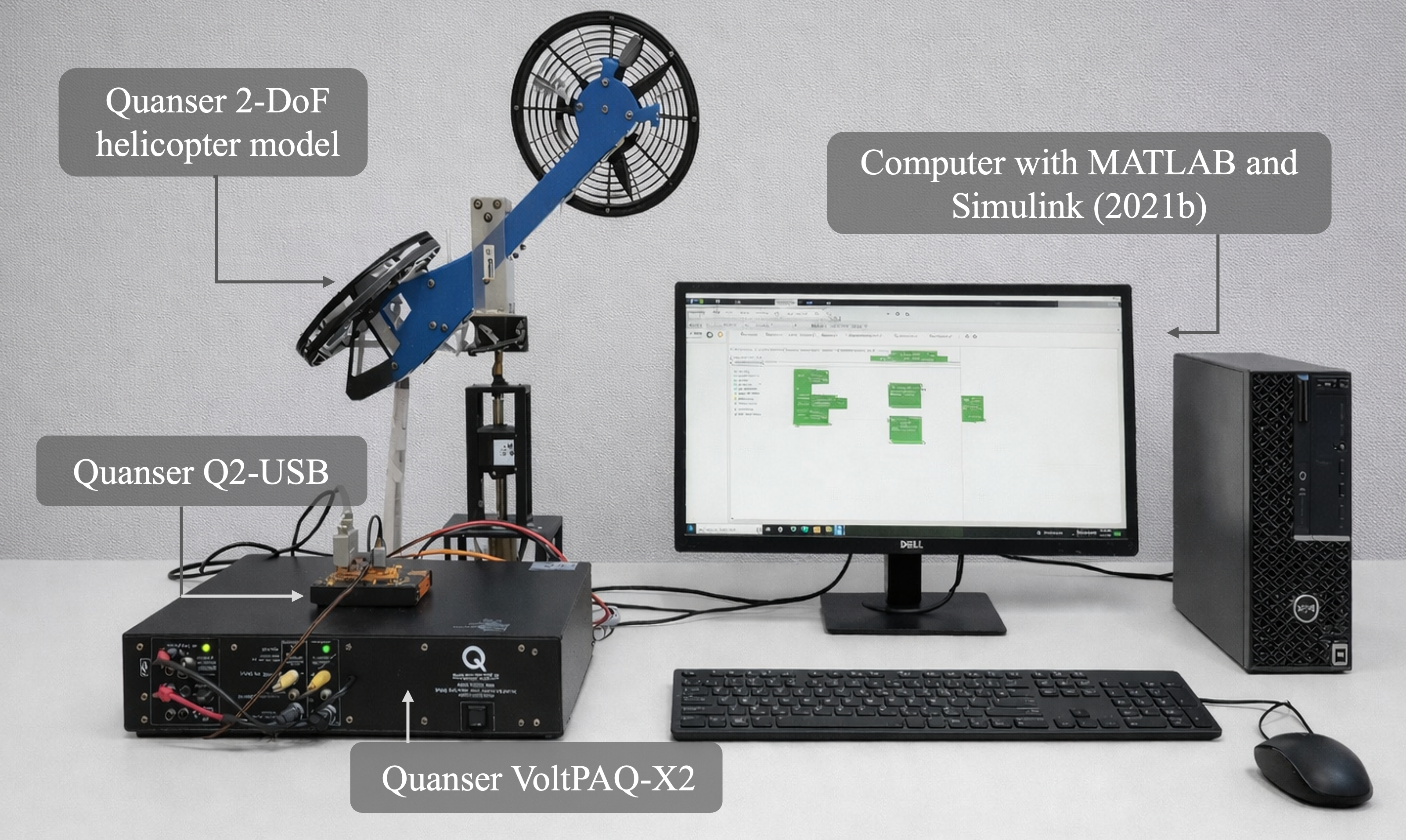}
    \caption{}
    \label{heli_setup}
    \end{subfigure}
    \begin{subfigure}[t]{0.48\textwidth}
    \centering
\includegraphics[width=0.85\linewidth]{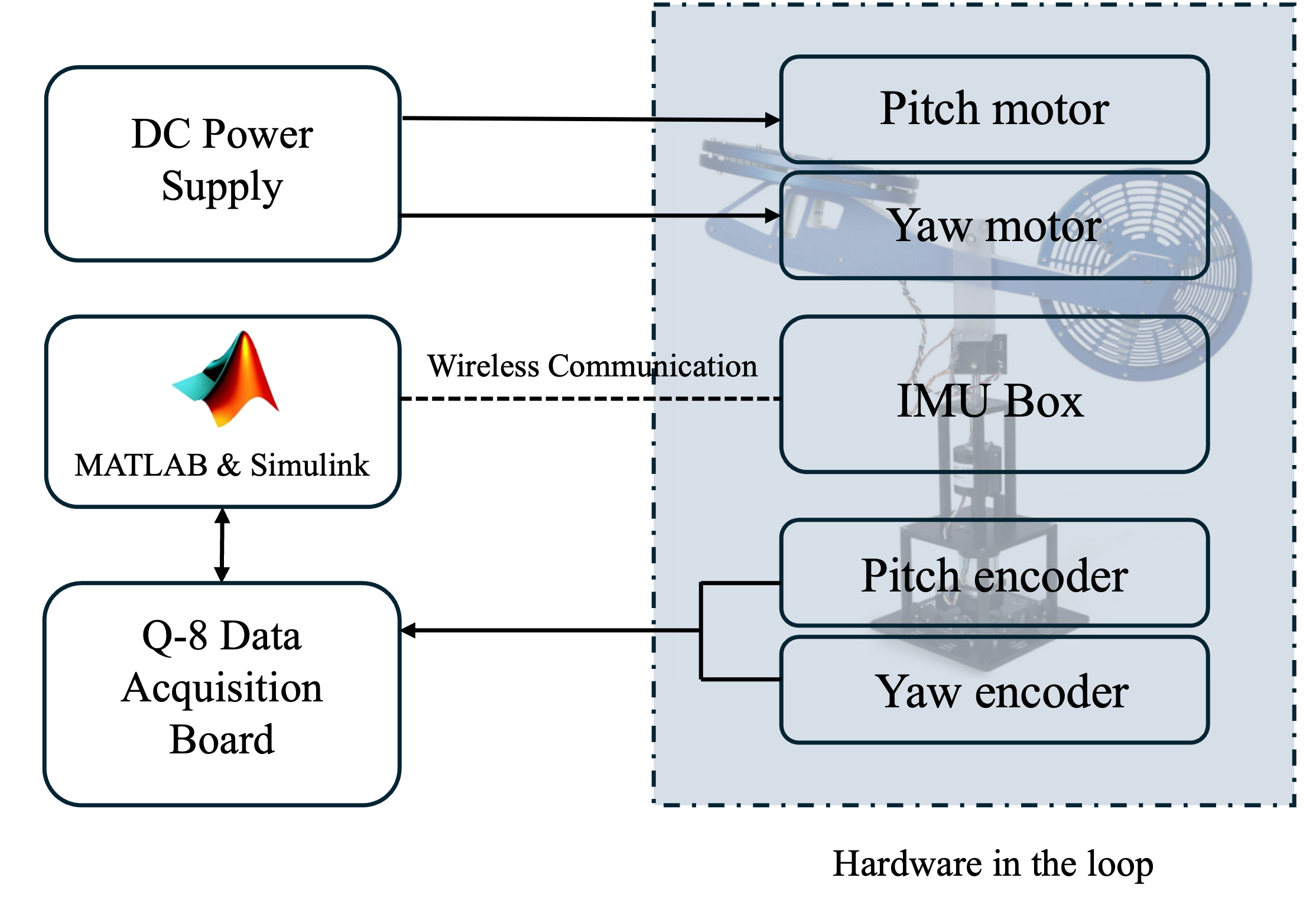}
    \caption{}
    \label{heli_block}
    \end{subfigure}
    \caption{(a) Experimental setup \iffalse \footnotemark $\;$ \fi and (b) simplified block diagram of the Quanser 2-DoF helicopter.}
    \label{}
\end{figure*}
\section{Designing Time-varying Constraints as Prescribed Performance Functions (PPF)}
\label{splcases}
\label{sec5}
After establishing closed-loop stability and constraint satisfaction in Section~\ref{PM}, we now present the systematic design of the time-varying state and input constraints. Specifically, we show that the proposed framework subsumes PPC and FC-based approaches as special cases by appropriately choosing prescribed performance functions (PPFs). This design flexibility allows shaping of the transient and steady-state behavior of both the tracking error and the control input while preserving constraint satisfaction.
A commonly adopted choice of the performance function is defined by
\begin{equation}
    \phi_{x}(t) \triangleq (\phi^{0}_{x} - \phi^{\infty}_{x} )\left(1 + \kappa_{x} t^{\nu_{x}}\right)^{-1} + \phi^{\infty}_{x}
    \label{ppf}
\end{equation}
where, \( \phi_{x}: [0, \infty) \to (\phi^{0}_{x}, \phi^{\infty}_{x}] \), \( \phi^{0}_{x} > 0 \) denotes the initial bound, ${\phi}^{\infty}_{x}>0$ represents the final bound and $\phi^{0}_{x}>\phi^{\infty}_{x}$. The parameters \( \kappa_{x}>0\), \(\nu_{x} >0 \) are user-defined positive constants that regulate the rate and shape of the performance envelope, respectively. Specifically, the parameter \( \kappa_{x} \) governs the decay rate of the performance function and, consequently, determines how rapidly the tracking error is driven toward its prescribed steady-state bound. Larger values of \( \kappa_{x} \) induce faster convergence, whereas smaller values lead to a more gradual reduction of the performance envelope. The parameter \( \nu_{x} \) modulates the shape of the convergence profile. For \( 0 < \nu_{x} < 1 \), the performance envelope exhibits a concave decay, characterized by a rapid initial contraction followed by a slower decay to the steady-state bound. In contrast, \( \nu_{x} > 1 \) yields a convex decay profile, resulting in slower initial convergence and faster steady-state convergence. The case \( \nu_{x} = 1 \) corresponds to an exponential-type decay.
The time at which the performance function reaches a prescribed accuracy level \( \epsilon_{x} \), with \( \epsilon_{x} > \phi^{\infty}_{x} \), can be explicitly determined by solving
\begin{align}
    \phi_{x}(t) = \epsilon_{x}
    \label{ppf_time}
\end{align}
Substituting the prescribed performance function in \eqref{ppf_time} yields
\begin{equation}
\label{ppf1}
\epsilon_{x}
=
(\phi^{0}_{x} - \phi^{\infty}_{x})
\left(1 + \kappa_{x} t^{\nu_{x}}\right)^{-1}
+ \phi^{\infty}_{x}.
\end{equation}
Solving \eqref{ppf1}, we obtain the convergence time as
\begin{equation}
\tau_{x}
=
\left[\frac{1}{\kappa_{x}}\left(
\frac{\phi^{0}_{x} - \phi^{\infty}_{x}}
{\epsilon_{x} - \phi^{\infty}_{x}}- 1\right)\right]^{\frac{1}{\nu_{x}}}
\label{pre-defined_time}
\end{equation}
which highlights a fundamental trade-off between the performance parameters \( \kappa_{x} \) and \( \nu_{x} \). For a fixed target convergence time \( \tau_{x} \), increasing \( \nu_{x} \) allows for a smaller value of \( \kappa_{x} \), resulting in a slower initial contraction of the performance envelope followed by a sharper reduction near steady state. Conversely, decreasing \( \nu_{x} \) requires a larger \( \kappa_{x} \), leading to a more aggressive initial decay of the performance bound. 
\begin{figure}[H]
    \centering
    \begin{subfigure}[t]{0.48\linewidth}
    \centering
\includegraphics[width=\linewidth]{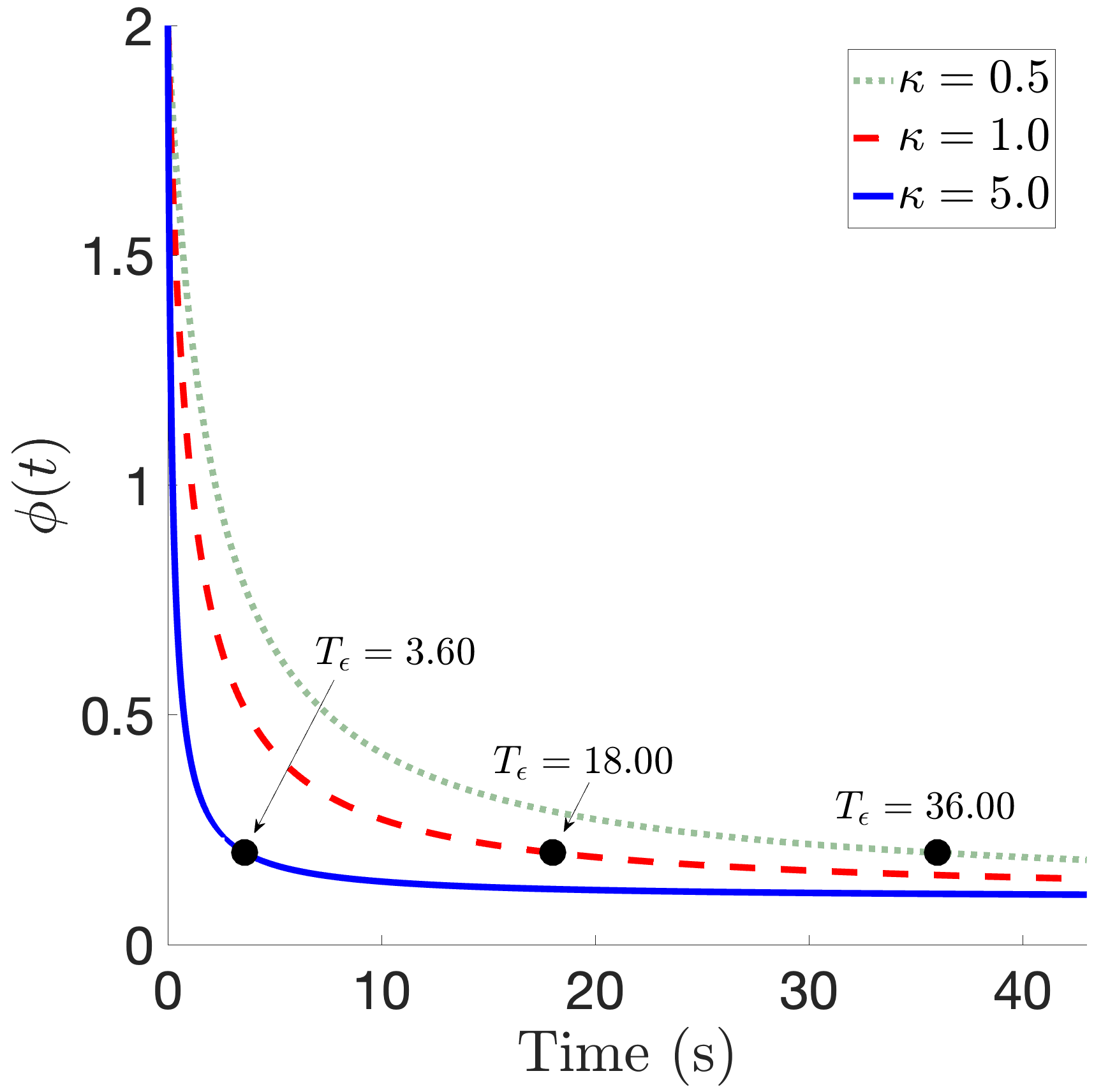}
    \caption{}
    \label{vary_kappa}
    \end{subfigure}
    \begin{subfigure}[t]{0.48\linewidth}
    \centering
\includegraphics[width=\linewidth]{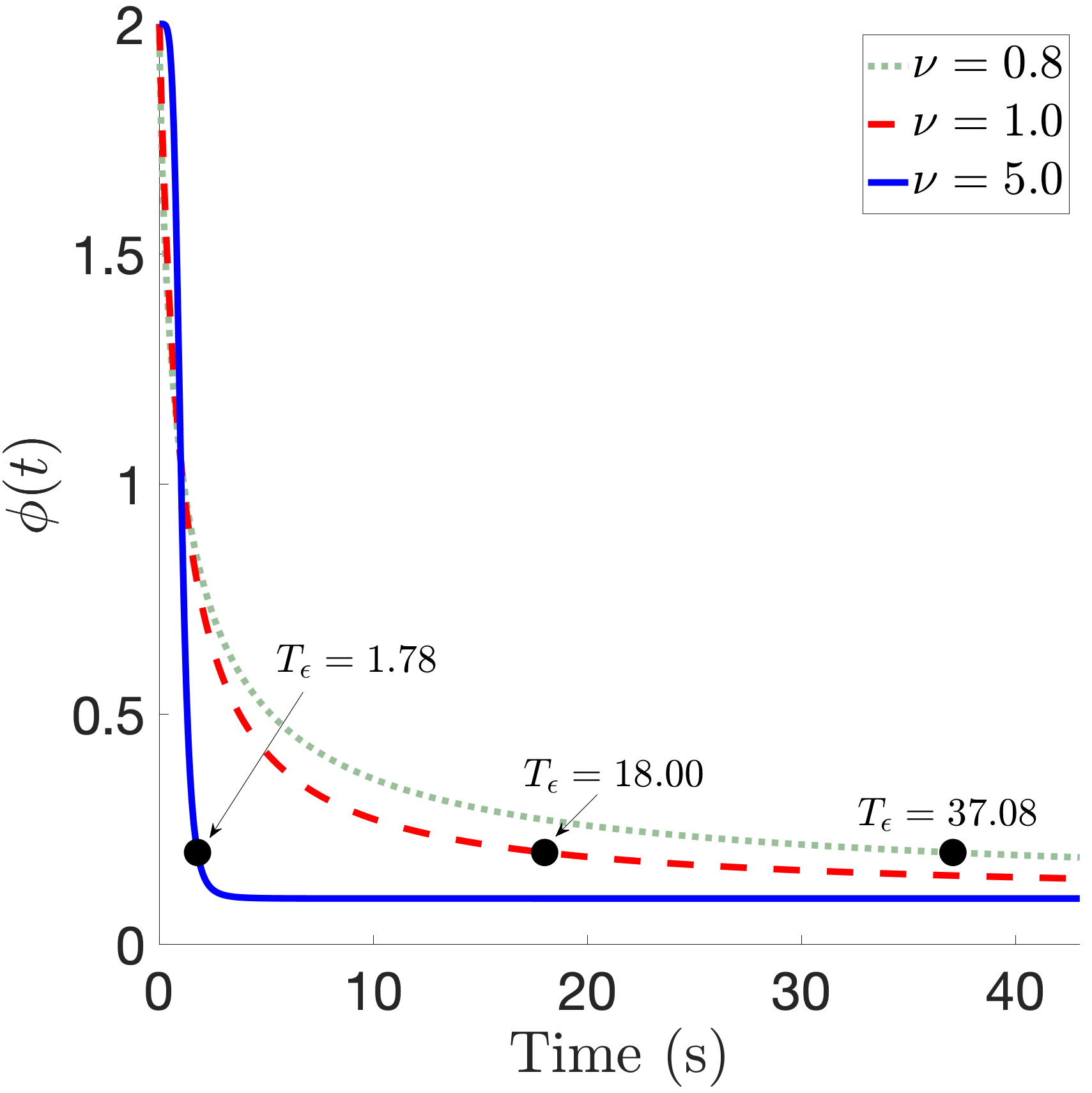}
    \caption{}
    \label{vary_nu}
    \end{subfigure}
    \caption{Effect of varying (a) $\kappa$  and (b) $\nu$ on the performance function $\phi_{x}(t)$, where $\phi_{x}^{0}=2$, $\phi_{x}^{\infty}=0.1$ and  $\epsilon=0.2$. $\nu=1$ and $\kappa=1$ are fixed in (a) and (b), respectively.}
    \label{vary_kappa_nu}
\end{figure}
\begin{remark}
    The convergence time \( \tau_{x} \) represents the tightening of the prescribed performance envelopes, rather than directly ensuring finite-time convergence of the trajectories themselves. Since the tracking error and the control input are guaranteed to remain within the corresponding pre-defined performance bounds for all \( t \ge 0 \), tuning \( \tau_{x} \) enables the designer to prescribe the time required for the constrained states to enter within a neighborhood of the steady-state limits.
\end{remark}
\begin{remark}
    In the proposed framework, the PPFs are employed exclusively to impose time-varying state and input constraints. As a result, the constraint design is decoupled from the adaptive law and feedback structure, allowing the designer to specify performance requirements independent of the controller parameterization. This is in contrast to classical PPC and FC formulations, where performance shaping is often intertwined with error transformations or controller gains. Moreover, while standard PPC or FC-based formulations enforce time-varying error bounds, the presence of input saturation is typically addressed either by online reshaping of the performance funnel \cite{fc3}, or by modifying the reference trajectory \cite{ppc3}. Consequently, imposing constraints directly on the plant states under input saturation is challenging using these frameworks. In contrast, our approach provides an offline verifiable feasibility condition that certifies a feasible set of time-varying constraints $(\phi_e(t), \phi_{\dot{e}}(t),\phi_{\tau}(t))$ for the given plant and reference model.
\end{remark}
\section{Experimental Results}
To illustrate the effectiveness of the proposed constrained control algorithm, real-time experiments are conducted on the Quanser 2-DoF helicopter model (Fig.~\ref{heli_block}). We consider a simplified dynamic model as follows. 
\begin{align}
\ddot{\theta}
&= \frac{1}{J_p+ml^2}\Big(
   \tau_{p}-B_p\dot{\theta}
    - \tfrac{1}{2} m l^2 \sin(2\theta)\,\dot{\psi}^2
    - m g l \cos\theta
   \Big) \\[4pt]
\ddot{\psi}
&= \frac{1}{J_y + m l^2 \cos^2\theta}\Big(
    \tau_y-B_y\dot{\psi}
    + m l^2 \sin(2\theta)\,\dot{\theta}\dot{\psi}
   \Big)
\end{align}
where, $\theta(t)\in\mathbb{R}$ (rad) and  $\psi(t)\in\mathbb{R}$  (rad) denote the pitch and yaw angles, respectively, while their time derivatives $\dot{\theta}(t)\in \mathbb{R}$ and $\dot{\psi}(t)\in\mathbb{R}$ (rad/sec)  represent the corresponding angular velocities. The generalized position and velocity vectors are defined as $q(t)=\begin{bmatrix}
    \theta(t) & \psi(t)
\end{bmatrix}^{\top}\in \mathbb{R}^2$ and $\dot{q}(t)=\begin{bmatrix}
    \dot{\theta}(t) & \dot{\psi}(t)
\end{bmatrix}^{\top}\in \mathbb{R}^2$, respectively.  \\
The experimental setup consists of two DC motors driving the front (main) and rear (tail) propellers. The main propeller predominantly governs the pitch motion, while the tail propeller influences the yaw motion.
Here, $\tau_p(t)$ and $\tau_y(t)$ denote the generalized control torques generated along the pitch and yaw axes, respectively, and the torque input vector is defined as 
$\tau(t)=\begin{bmatrix}\tau_p(t) & \tau_y(t)\end{bmatrix}^{\top}\in\mathbb{R}^2$. These torques are produced by the thrust forces generated by the propellers, which are actuated via the motor voltage $V(t)\in\mathbb{R}^2$. The mapping between the applied motor voltages and the resulting generalized torques is assumed to be linear and is given by 
\begin{align}
\tau(t)=T\,V(t)
\end{align}
where,
\begin{align}
T= 
\begin{bmatrix}
K_{pp} & K_{py}\\[2pt]
K_{yp} & K_{yy}\
\end{bmatrix}
\end{align} 
is a nonsingular thrust-torque conversion matrix. The system parameters and their numerical values are summarized in Table~\ref{tab:params_compact}. The pitch and yaw angles are measured using high-resolution optical encoders. The angular velocities are obtained via numerical differentiation of the encoder signals.  As the propeller angular velocities are not directly measurable, the actuator dynamics and gyroscopic coupling effects are neglected in the model derivation. 
The desired reference trajectory is given by
\begin{equation}
   q_d(t)=\begin{bmatrix}
    -34+ 2\sin(0.5t) \\ 0
    \end{bmatrix}
\end{equation}
The objective is to design an adaptive tracking controller such that the tracking errors $e(t)$ and $\dot{e}(t)$ remain bounded within a known dynamic envelope while simultaneously satisfying user-defined time-varying constraints on the state and input. Under Assumption~\ref{ref_model_assumption}, constraints imposed on the system state can be equivalently expressed as constraints on the tracking error, and vice versa. For simulation, the constraints are specified in terms of the tracking error and its derivative; however, unlike PPC and FC-based approaches, the proposed method allows these constraints to be directly interpreted as constraints on the system states.
 The time-varying bounds on the errorand control input are considered as
\begin{align}
    &\phi_e(t)=(\phi_e^0-\phi_e^{\infty})(1+\kappa_e t^{\nu_e})^{-1}+\phi_e^{\infty} \label{phie1}\\
     &\phi_{\dot{e}}(t)=(\phi_{\dot{e}}^0-\phi_{\dot{e}}^{\infty})(1+\kappa_{\dot{e}} t^{\nu_{\dot{e}}})^{-1}+\phi_{\dot{e}}^{\infty} \label{phiedot1}\\
    &\phi_{\tau}(t)=(\phi_{\tau}^0-\phi_{\tau}^{\infty})(1+\kappa_{\tau} t^{\nu_{\tau}})^{-1}+\phi_{\tau}^{\infty}
    \label{phiu1}
\end{align}
where, $\phi^0_{x},\phi_{x}^{\infty},\kappa_{x},\nu_{x}$ for $x=e,\dot{e},\tau$ are defined in Section \ref{sec5}. We have considered $\phi_e^0=11,\:  \phi_e^{\infty}=1, \phi_{\dot{e}}^0=4.5,\:  \phi_{\dot{e}}^{\infty}=1.5, \: \phi_{\tau}^0=6,\:\phi^{\infty}_{\tau}=5,\: \:\kappa_e=\kappa_{\tau}=0.2,\: \kappa_{\dot{e}}=0.1,\: \nu_e=\nu_{\dot{e}}=\nu_{\tau}=1$.
 The other parameters are chosen as: $\Gamma=2\mathbb{I}_{m\times m}$, 
$\alpha_3=0.1$, $\bar{d}=0.5$, $K=\begin{bmatrix}
    1.5 & 0\\
    0 & 0.2
\end{bmatrix}$, $\bar{\theta}=0.91$, $k_{m_{2}}=0.0908$, $k_v=0.03365$, $k_g=2.514$, $k_{f_{1}}=0$, and $k_{f_{2}}=0.8$. The constraints and design parameters are chosen such that they satisfy the feasibility condition C1.\\
In practice, the bound $\bar{\theta}$ can be obtained from prior knowledge of the physical parameters of the helicopter system. Conservative bounds on quantities such as the inertia, thrust coefficients, and motor constants are typically available from manufacturer specifications or experimental identification. These bounds can be used to derive an upper bound on the unknown parameter vector $\theta$. The resulting bound may be conservative; however, such conservatism is a common trade-off that ensures the feasibility condition to be verified a priori under pre-defined constraints and parametric uncertainties.\\
Fig.~\ref{pitch_angle} and Fig.~\ref{yaw_angle} show that the pitch and yaw angles track the corresponding reference trajectories. We show in Fig.~\ref{position_error} that the position tracking error remain within the predefined safe set, i.e., $\|e(t)\|<\phi_e(t)$ $\forall t \ge 0$. Similarly, the pitch and yaw velocities also track the desired trajectories (Fig.~\ref{pitch_velocity}, \ref{yaw_velocity}) and the velocity error satisfies the prescribed dynamic constraint, i.e., $\|\dot{e}(t)\|<\phi_{\dot{e}}(t)$ $\forall t \ge 0$ as shown in Fig.~\ref{velocity_error}. Further, the proposed controller ensures that the required control input remains within the user-defined time-varying constraint, as shown in Fig.~\ref{input_sim}. \\
Although the proposed design prescribes time-varying state and input constraints on the plant, in the experimental study, we have designed the time-varying error envelopes as prescribed performance functions (PPFs) to ensure that the trajectory tracking errors \(\|e(t)\|\) and \(\|\dot e(t)\|\) remain within these envelopes. Under Assumption~\ref{ref_model_assumption}, these error bounds can be directly translated into time-varying constraints on the plant states, i.e., \(\|q(t)\|<\phi_e(t)+34-2\sin(0.5 t)\) $\forall t \ge 0$ and \(\|\dot q(t)\|<\phi_{\dot{e}}+\cos(0.5t)\) $\forall t \ge 0$.\\ 
Equivalently, we may prescribe the desired time-varying state constraints and employing Lemma~\ref{lemma_constraint_conv}, translate them into admissible bounds on the trajectory tracking error, followed by the filtered tracking error. The proposed controller and adaptive laws then ensure forward invariance of the resulting error sets and, consequently, satisfaction of the original state and input constraints.
\begin{remark}
    A direct quantitative comparison with standard PPC or FC-based methods is not straightforward, since those approaches typically accommodate input saturation by modifying the reference trajectory or reshaping the performance envelope online, which is explicitly excluded in the present framework.
\end{remark}

\begin{figure}[H]
    \centering
    \includegraphics[width=\linewidth]{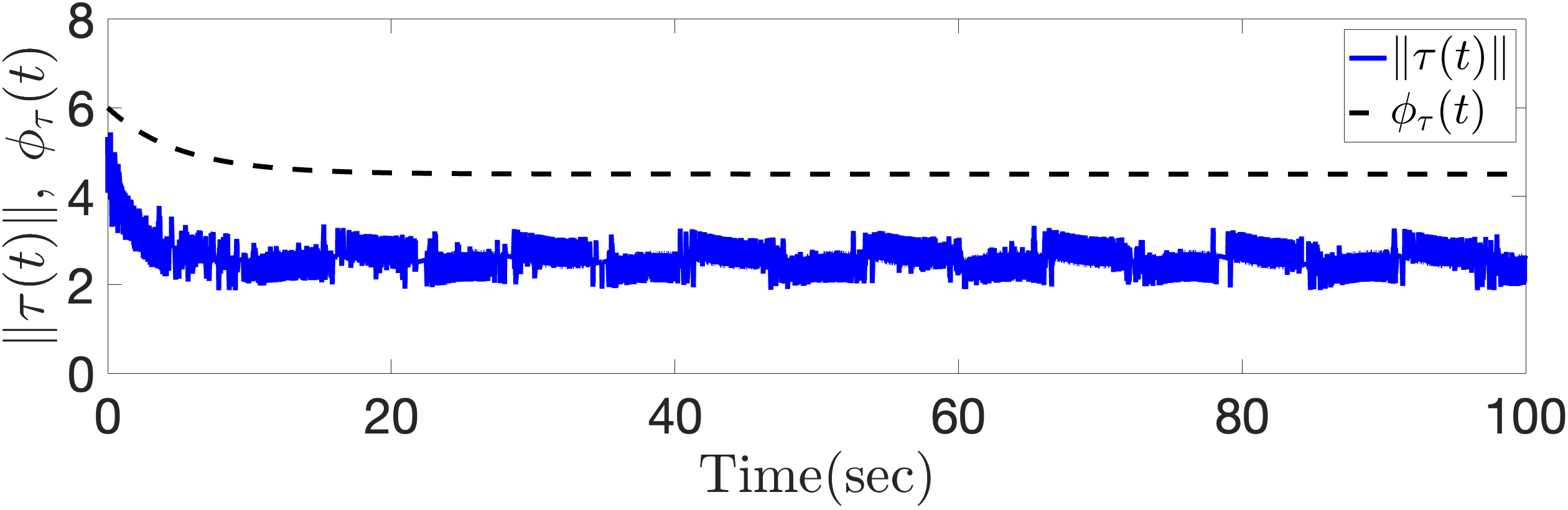}
    \caption{Required control input for the 2-DoF helicopter using the proposed controller.}
    \label{input_sim}
\end{figure}
\begin{table}[htbp]
\centering
\scriptsize
\setlength{\tabcolsep}{4pt}
\renewcommand{\arraystretch}{1.1}
\caption{Value of the model parameters with significance and units.}
\label{tab:params_compact}
\begin{tabular}{@{} l >{\raggedright\arraybackslash}p{0.72\columnwidth} >{\raggedright\arraybackslash}p{0.12\columnwidth} @{}}
\toprule
\textbf{Symbol} & \textbf{Significance (units)} & \textbf{Value} \\
\midrule
$J_p$      & Moment of inertia about pitch axis (kg/m$^2$) & $0.0384$ \\
$J_y$      & Moment of inertia about yaw axis (kg/m$^2$) & $0.0432$ \\
$m$ & Total mass of helicopter(Kg) & $1.38$\\
$l$        & Pitch–yaw lever distance (m) & $0.1857$ \\
$B_p$ & Viscous damping constant of pitch (N.m.s/rad) & 0.8\\
 $B_y$ & Viscous damping constant of yaw (N.m.s/rad) & 0.318\\
$K_{pp}$   & Torque constant acting on pitch axis from pitch motor (N·m/V) & $0.2041$\\
$K_{yy}$   & Torque constant acting on yaw axis from yaw motor (N·m/V) & $0.0720$ \\
$K_{py}$   & Torque constant acting on pitch axis from yaw motor (N·m/V) & $0.0068$ \\
$K_{yp}$   & Torque constant acting on yaw axis from pitch motor  (N·m/V) & $0.0219$ \\
$g$        & Gravitational acceleration (m/s$^2$) & 9.81 \\
\bottomrule
\end{tabular}
\end{table}
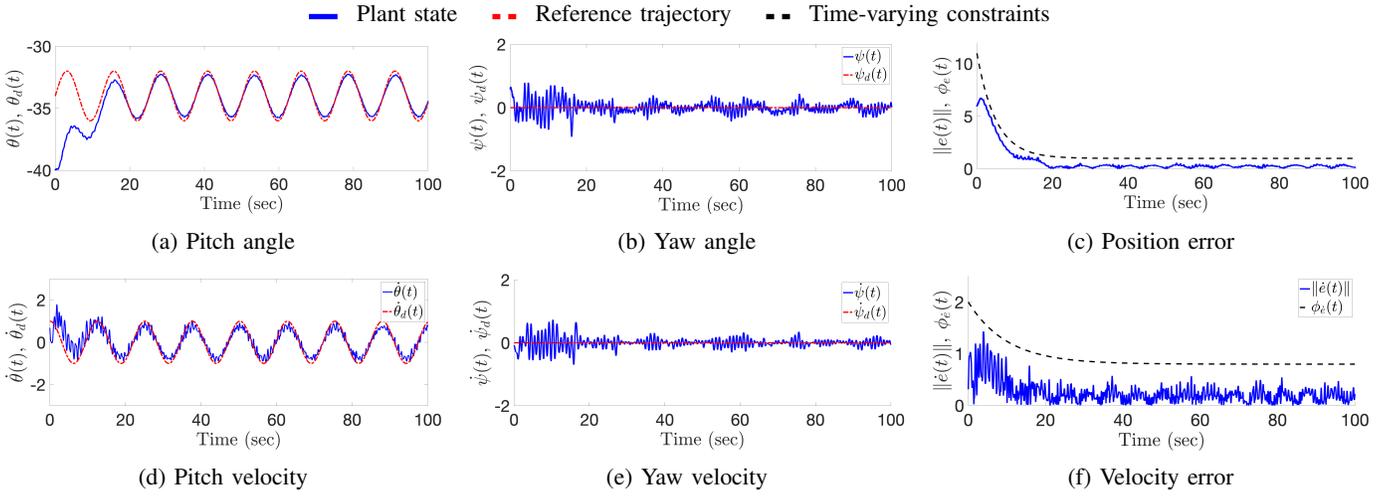
\begin{figure*}[t]
\centering
{\small
\begin{tikzpicture}[baseline]
\node[inner sep=0pt] (txt) {
\textcolor{blue}{\tikz{\draw[line width=2pt] (0,0.2em) -- (1.2em,0.2em);}}
~Plant state \quad
\textcolor{red}{\tikz{\draw[line width=2pt,dashed] (0,0.2em) -- (1em,0.2em);}}
~Reference trajectory \quad
\textcolor{black}{\tikz{\draw[line width=2pt,dashed] (0,0.2em) -- (1em,0.2em);}}
~Time-varying constraints
};
\end{tikzpicture}
}
\vspace{0.4em}

\begin{subfigure}[t]{0.32\textwidth}
    \centering
    \includegraphics[width=\linewidth]{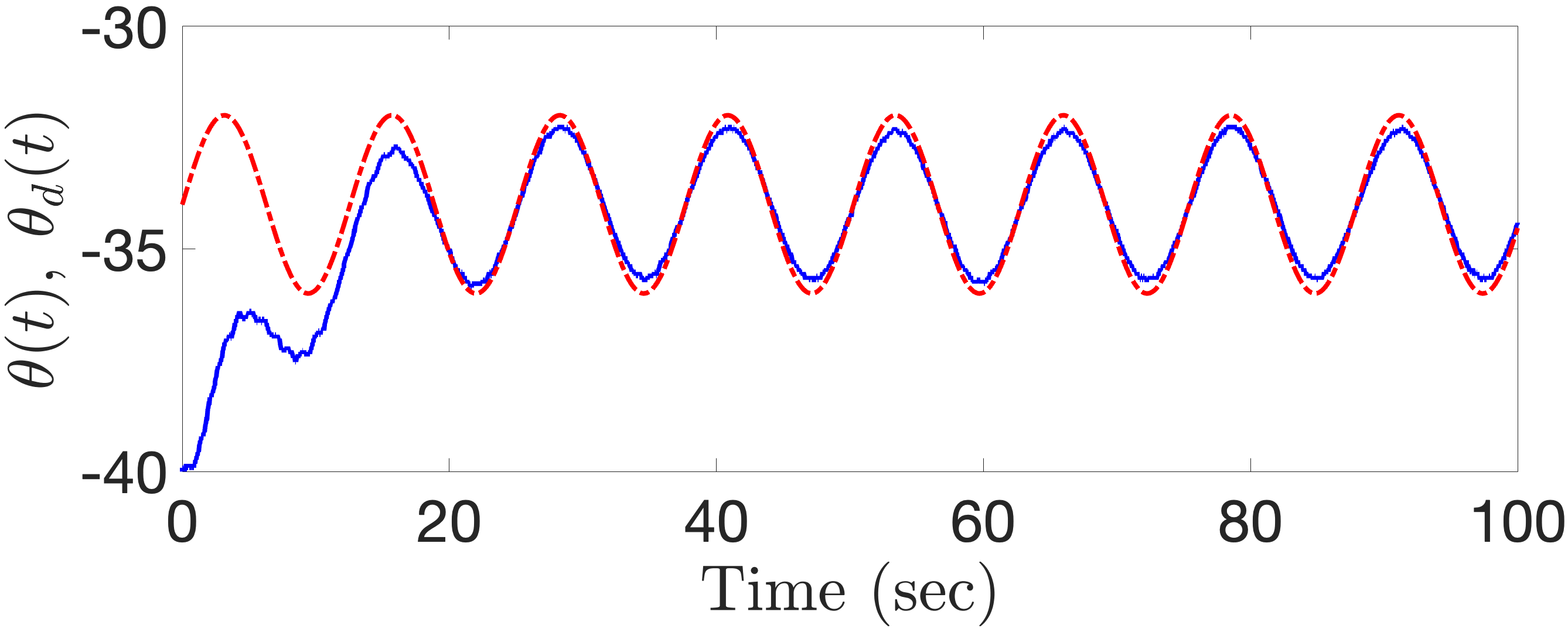}
    \caption{Pitch angle}
    \label{pitch_angle}
\end{subfigure}\hfill
\begin{subfigure}[t]{0.32\textwidth}
    \centering
    \includegraphics[width=\linewidth]{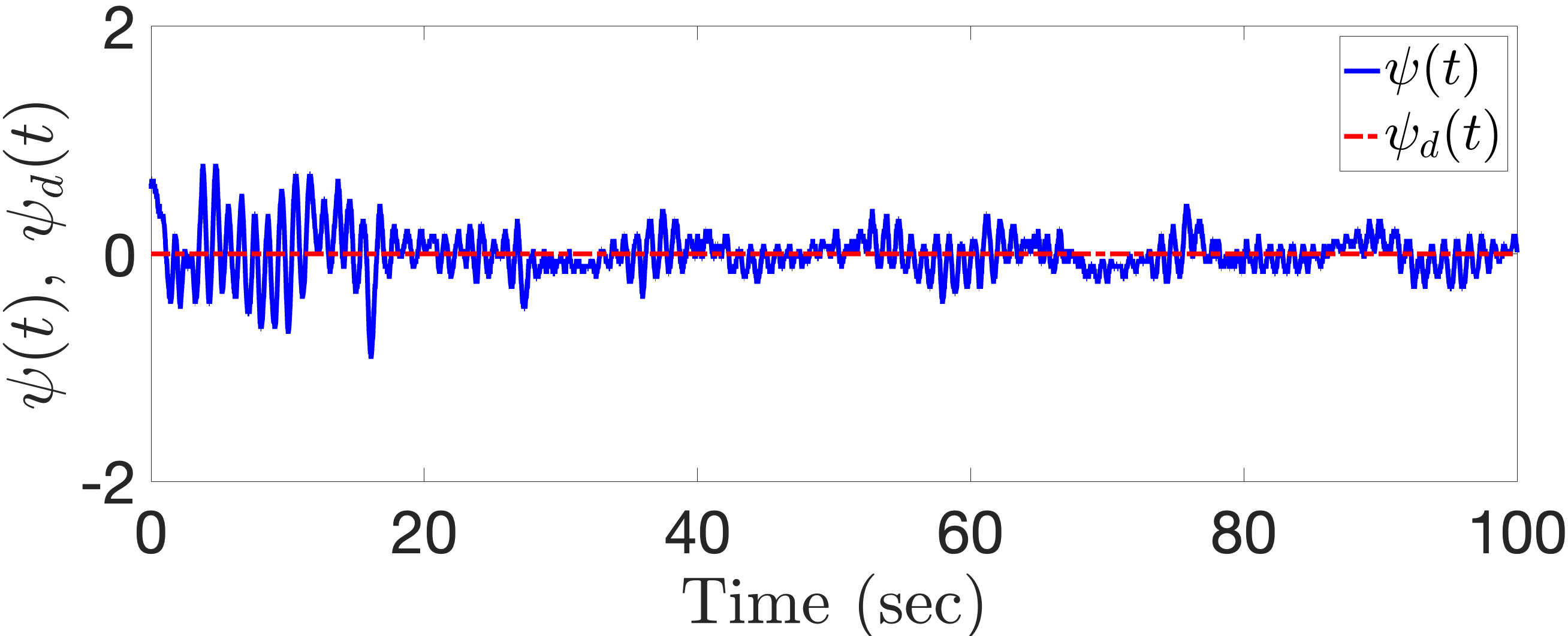}
    \caption{Yaw angle}
    \label{yaw_angle}
\end{subfigure}\hfill
\begin{subfigure}[t]{0.32\textwidth}
    \centering
    \includegraphics[width=\linewidth]{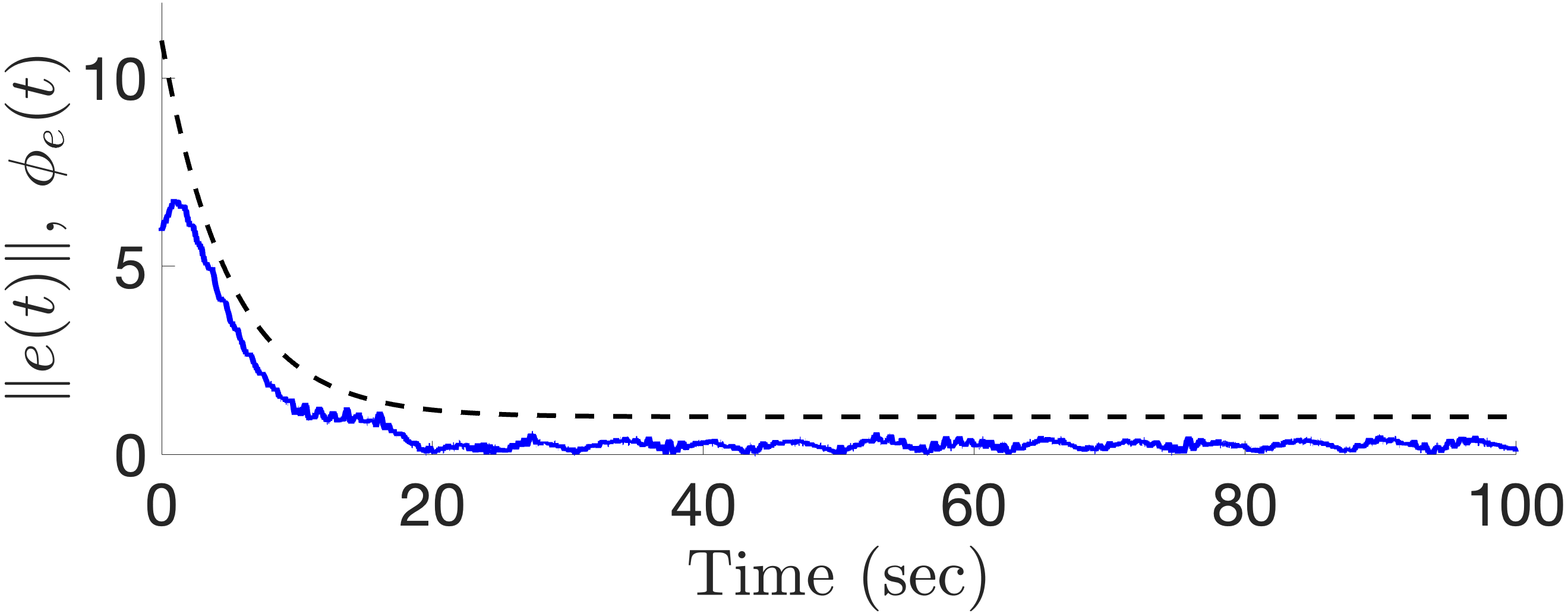}
    \caption{Position error}
    \label{position_error}
\end{subfigure}

\vspace{0.6em}

\begin{subfigure}[t]{0.32\textwidth}
    \centering
    \includegraphics[width=\linewidth]{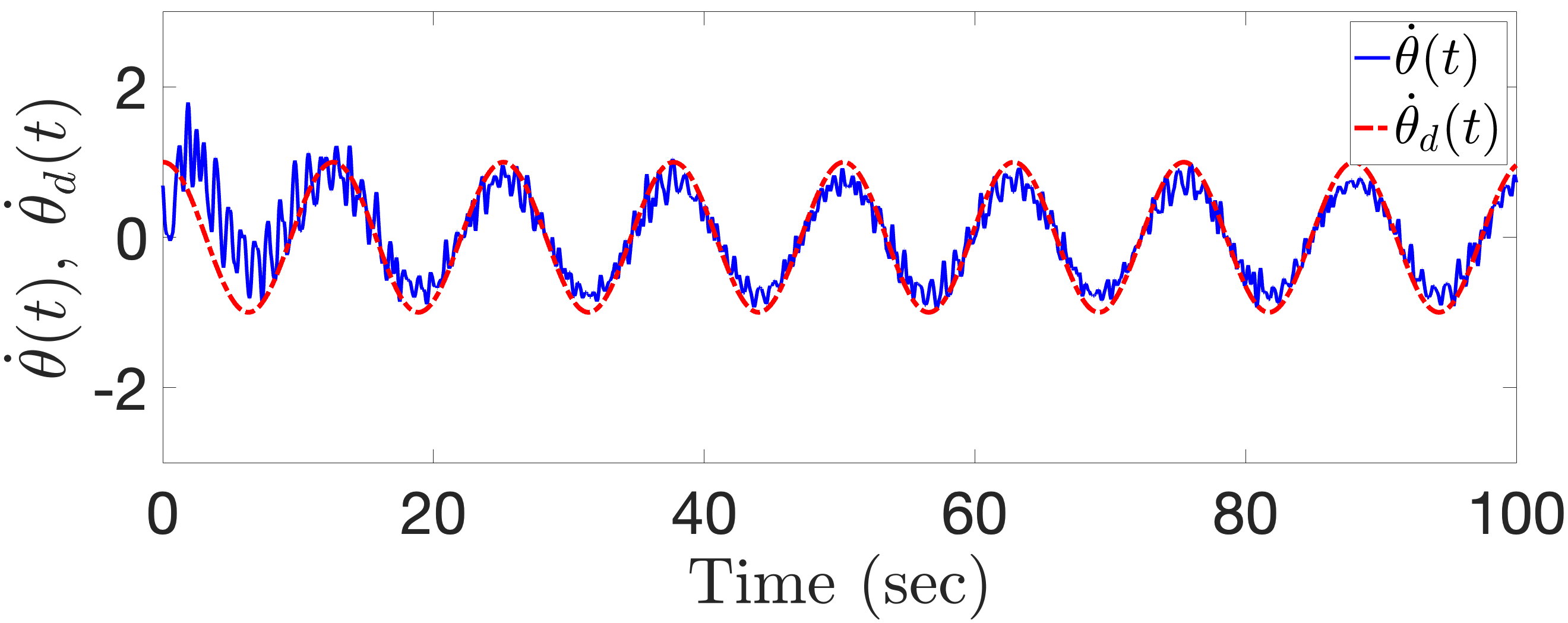}
    \caption{Pitch velocity}
    \label{pitch_velocity}
\end{subfigure}\hfill
\begin{subfigure}[t]{0.32\textwidth}
    \centering
    \includegraphics[width=\linewidth]{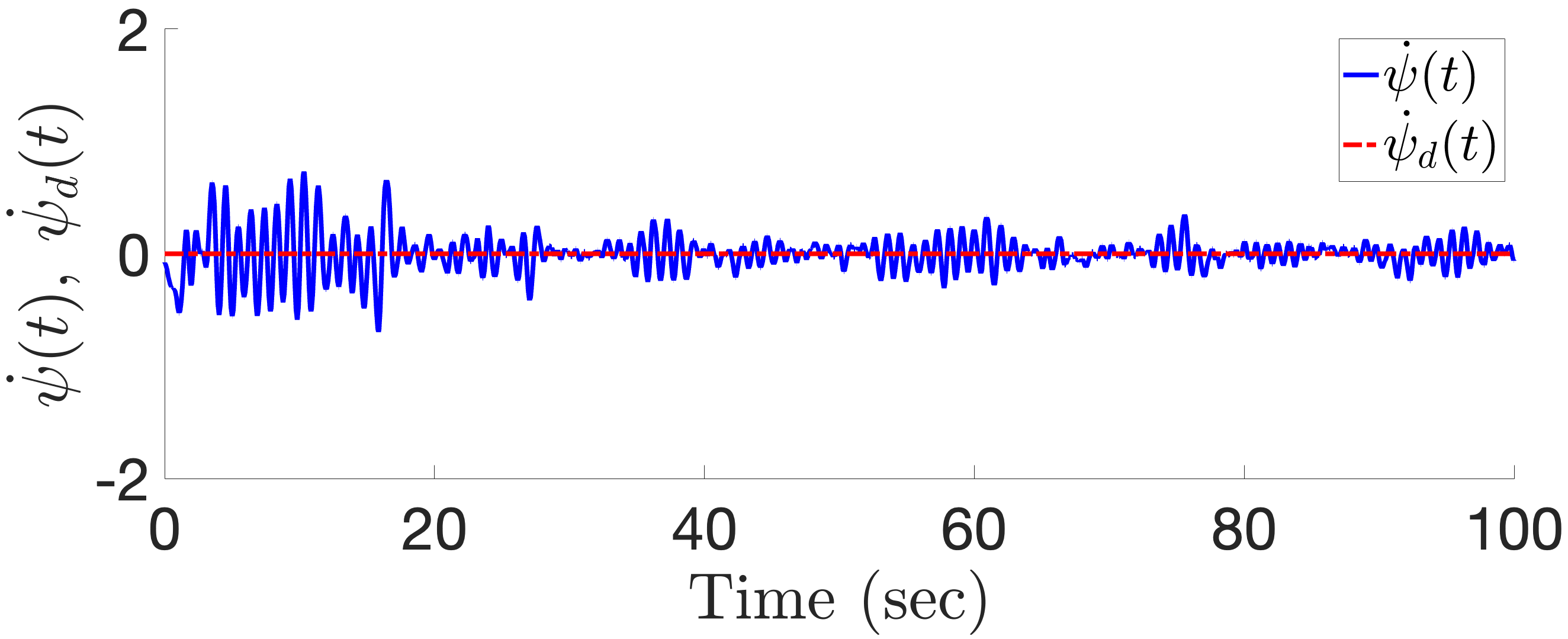}
    \caption{Yaw velocity}
    \label{yaw_velocity}
\end{subfigure}\hfill
\begin{subfigure}[t]{0.32\textwidth}
    \centering
    \includegraphics[width=\linewidth]{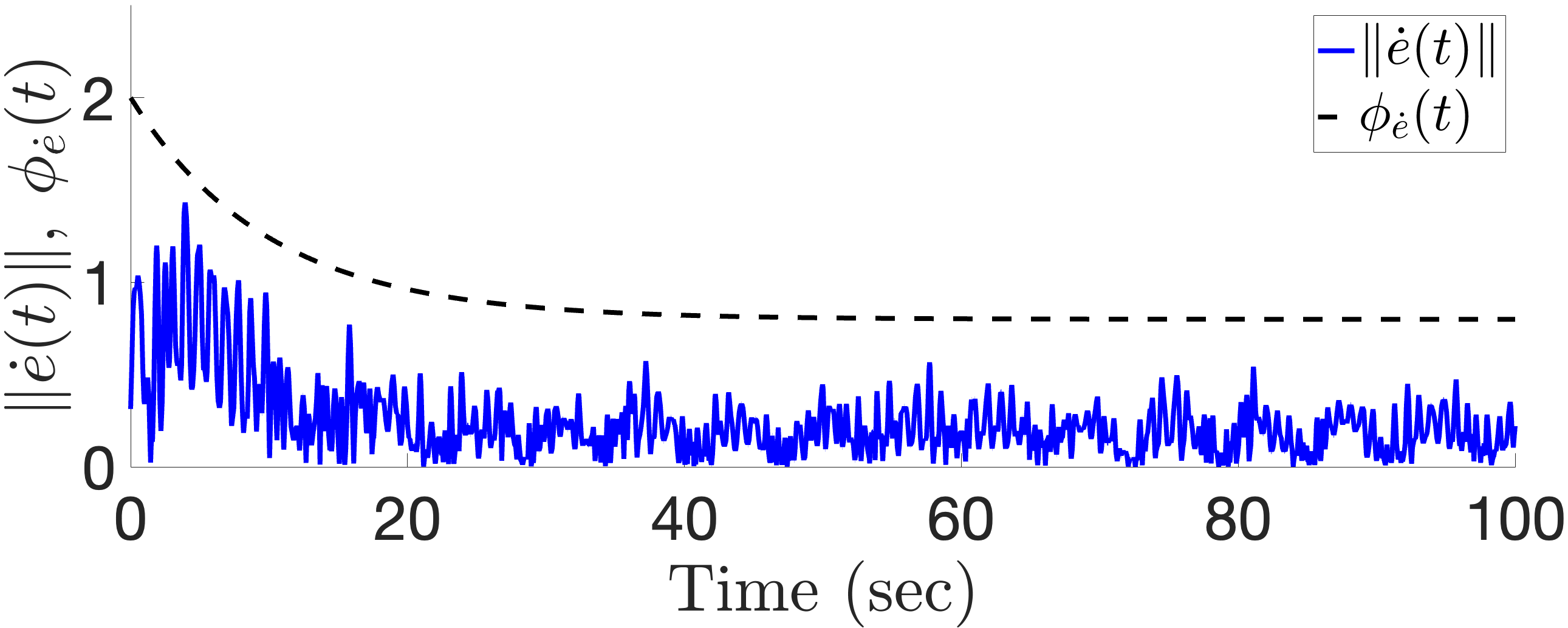}
    \caption{Velocity error}
    \label{velocity_error}
\end{subfigure}
\caption{Experimental results using the proposed controller.}
\label{fig:states_errors_2x3}
\end{figure*}
\section{Conclusion}
In this paper, we design an optimization-free safe adaptive tracking control algorithm for E-L systems, which ensures time-varying state and input constraint satisfaction in the presence of parametric uncertainties and bounded external disturbances. The control strategy integrates a time-varying barrier Lyapunov function (TVBLF)-based approach with a saturated controller, ensuring that both plant states and control inputs remain within predefined dynamic envelopes. We also guarantee boundedness of tracking error within known time-varying limits. We show that through the design of the safe envelopes as performance functions,  PPC and FC-based methods are subsumed as special cases of this more general approach, providing a viable alternative in cases where the reference or performance envelope modification is not permitted. Furthermore, an offline verifiable feasibility condition is provided to guarantee the feasibility of the control policy. Experimental results on Quanser 2-DoF helicopter demonstrate the efficacy of the proposed control scheme.

\bibliographystyle{ieeetr}
\bibliography{ref}

@book{Vidyasagar,
  title={Robot dynamics and control},
  author={Spong, Mark W and Vidyasagar, Mathukumalli},
  year={2008},
  publisher={John Wiley \& Sons}
}

@ARTICLE{ppc1,
  author={Bechlioulis, Charalampos P. and Rovithakis, George A.},
  journal={IEEE Transactions on Automatic Control}, 
  title={Robust Adaptive Control of Feedback Linearizable MIMO Nonlinear Systems With Prescribed Performance}, 
  year={2008},
  volume={53},
  number={9},
  pages={2090-2099},
  doi={10.1109/TAC.2008.929402}}

@article{ppc2,
  title={Prescribed performance adaptive control for multi-input multi-output affine in the control nonlinear systems},
  author={Bechlioulis, Charalampos P and Rovithakis, George A},
  journal={IEEE Transactions on automatic control},
  volume={55},
  number={5},
  pages={1220--1226},
  year={2010},
  publisher={IEEE}
}

@article{fc1,
  title={Funnel control of nonlinear systems},
  author={Berger, Thomas and Ilchmann, Achim and Ryan, Eugene P},
  journal={Mathematics of Control, Signals, and Systems},
  volume={33},
  pages={151--194},
  year={2021},
  publisher={Springer}
}

@article{fc2,
  title={Funnel control for a class of nonlinear infinite-dimensional systems},
  author={Hastir, Anthony and Winkin, Joseph J and Dochain, Denis},
  journal={Automatica},
  volume={152},
  pages={110964},
  year={2023},
  publisher={Elsevier}
}

@article{ppc3,
  title={Input-constrained prescribed performance control for high-order MIMO uncertain nonlinear systems via reference modification},
  author={Fotiadis, Filippos and Rovithakis, George A},
  journal={IEEE Transactions on Automatic Control},
  volume={69},
  number={5},
  pages={3301--3308},
  year={2023},
  publisher={IEEE}
}

@article{fc3,
  title={Input-constrained funnel control of nonlinear systems},
  author={Berger, Thomas},
  journal={IEEE Transactions on Automatic Control},
  volume={69},
  number={8},
  pages={5368--5382},
  year={2024},
  publisher={IEEE}
}

@article{funnelsaturation1,
  title={Funnel control with saturation: Linear MIMO systems},
  author={Hopfe, Norman and Ilchmann, Achim and Ryan, Eugene P},
  journal={IEEE Transactions on Automatic Control},
  volume={55},
  number={2},
  pages={532--538},
  year={2010},
  publisher={IEEE}
}

@inproceedings{zcbf1,
  title={Correct-by-design control barrier functions for euler-lagrange systems with input constraints},
  author={Cortez, Wenceslao Shaw and Dimarogonas, Dimos V},
  booktitle={American Control Conference},
  pages={950--955},
  year={2020},
  organization={}
}

@article{breeden2023robust,
  title={Robust control barrier functions under high relative degree and input constraints for satellite trajectories},
  author={Breeden, Joseph and Panagou, Dimitra},
  journal={Automatica},
  volume={155},
  pages={111109},
  year={2023},
  publisher={Elsevier}
}

@INPROCEEDINGS{icra_cbf,
  author={So, Oswin and Serlin, Zachary and Mann, Makai and Gonzales, Jake and Rutledge, Kwesi and Roy, Nicholas and Fan, Chuchu},
  booktitle={2024 IEEE International Conference on Robotics and Automation (ICRA)}, 
  title={How to Train Your Neural Control Barrier Function: Learning Safety Filters for Complex Input-Constrained Systems}, 
  year={2024},
  volume={},
  number={},
  pages={11532-11539},
  doi={10.1109/ICRA57147.2024.10610418}}

@incollection{camacho2007constrained,
  title={Constrained model predictive control},
  author={Camacho, Eduardo F and Bordons, Carlos},
  booktitle={Model predictive control},
  pages={177--216},
  year={2007},
  publisher={Springer}
}

@article{zheng1995stability,
  title={Stability of model predictive control with mixed constraints},
  author={Zheng, Alex and Morari, Manfred},
  journal={IEEE Transactions on automatic control},
  volume={40},
  number={10},
  pages={1818--1823},
  year={1995},
  publisher={IEEE}
}

@inproceedings{salehi2020safe,
  title={Safe tracking control of an uncertain euler-lagrange system with full-state constraints using barrier functions},
  author={Salehi, Iman and Rotithor, Ghananeel and Trombetta, Daniel and Dani, Ashwin P},
  booktitle={2020 59th IEEE Conference on Decision and Control (CDC)},
  pages={3310--3315},
  year={2020},
  organization={IEEE}
}

@book{li2018stability,
  title={Stability and performance of control systems with actuator saturation},
  author={Li, Yuanlong and Lin, Zongli},
  year={2018},
  publisher={Springer}
}

@article{roy2017adaptive,
  title={Adaptive--robust control of Euler--Lagrange systems with linearly parametrizable uncertainty bound},
  author={Roy, Spandan and Roy, Sayan Basu and Kar, Indra Narayan},
  journal={IEEE Transactions on Control Systems Technology},
  volume={26},
  number={5},
  pages={1842--1850},
  year={2017},
  publisher={IEEE}
}

@article{tvblf1,
  title={Control of nonlinear systems with time-varying output constraints},
  author={Tee, Keng Peng and Ren, Beibei and Ge, Shuzhi Sam},
  journal={Automatica},
  volume={47},
  number={11},
  pages={2511--2516},
  year={2011},
  publisher={Elsevier}
}

@article{elblf1,
  title={Prescribed performance control of uncertain Euler--Lagrange systems subject to full-state constraints},
  author={Zhao, Kai and Song, Yongduan and Ma, Tiedong and He, Liu},
  journal={IEEE Transactions on Neural Networks and Learning Systems},
  volume={29},
  number={8},
  pages={3478--3489},
  year={2017},
  publisher={IEEE}
}

@ARTICLE{ghoshtac,
  author={Ghosh, Poulomee and Bhasin, Shubhendu},
  journal={IEEE Transactions on Automatic Control}, 
  title={State and Input Constrained Model Reference Adaptive Control With Robustness and Feasibility Analysis}, 
  year={2026},
  volume={},
  number={},
  pages={1-7},
  doi={10.1109/TAC.2026.3654318}}

@article{tee2011control,
  title={Control of nonlinear systems with time-varying output constraints},
  author={Tee, Keng Peng and Ren, Beibei and Ge, Shuzhi Sam},
  journal={Automatica},
  volume={47},
  number={11},
  pages={2511--2516},
  year={2011},
  publisher={Elsevier}
}

@incollection{lavretsky2012robust,
  title={Robust adaptive control},
  author={Lavretsky, Eugene and Wise, Kevin A},
  booktitle={Robust and adaptive control: With aerospace applications},
  pages={317--353},
  year={2012},
  publisher={Springer}
}

@INPROCEEDINGS{myel1,
  author={Ghosh, Poulomee and Bhasin, Shubhendu},
  booktitle={2023 American Control Conference (ACC)}, 
  title={Adaptive Tracking Control of Uncertain Euler-Lagrange Systems with State and Input Constraints}, 
  year={2023},
  volume={},
  number={},
  pages={4229-4234},
  doi={10.23919/ACC55779.2023.10156436}}

@article{BLF,
  title={Barrier Lyapunov functions for the control of output-constrained nonlinear systems},
  author={Tee, Keng Peng and Ge, Shuzhi Sam and Tay, Eng Hock},
  journal={Automatica},
  volume={45},
  number={4},
  pages={918--927},
  year={2009},
  publisher={Elsevier}
}

@article{BLF2,
  title={Barrier Lyapunov functions-based adaptive control for a class of nonlinear pure-feedback systems with full state constraints},
  author={Liu, Yan-Jun and Tong, Shaocheng},
  journal={Automatica},
  volume={64},
  pages={70--75},
  year={2016},
  publisher={Elsevier}
}

@article{mpc11,
  title={Constrained model predictive control: Stability and optimality},
  author={Mayne, David Q and Rawlings, James B and Rao, Christopher V and Scokaert, Pierre OM},
  journal={Automatica},
  volume={36},
  number={6},
  pages={789--814},
  year={2000},
  publisher={Elsevier}
}

@book{slotine,
  title={Applied nonlinear control},
  author={Slotine, Jean-Jacques E and Li, Weiping and others},
  volume={199},
  year={1991},
  publisher={Prentice hall Englewood Cliffs, NJ}
}

@INPROCEEDINGS{arc1,
  author={Xiangbin Liu and Hongye Su and Bin Yao and Jian Chu},
  booktitle={2008 47th IEEE Conference on Decision and Control}, 
  title={Adaptive robust control of a class of uncertain nonlinear systems with unknown sinusoidal disturbances}, 
  year={2008},
  volume={},
  number={},
  pages={2594-2599},
  doi={10.1109/CDC.2008.4739272}}
\end{document}